\documentclass[a4paper,UKenglish,cleveref]{lipics-v2021}
\usepackage[T1]{fontenc}
\usepackage[utf8]{inputenc}
\usepackage{amsmath,amsfonts,amssymb,amsthm}
\usepackage{graphicx}
\usepackage{cleveref}
\usepackage[dvipsnames,svgnames]{xcolor}
\usepackage{mathtools}
\usepackage{enumerate}
\usepackage{hyphenat}
\usepackage{enumitem}
\usepackage[disable]{todonotes}
\usepackage{marginnote}
\hypersetup{hidelinks,breaklinks}

\usepackage{mathcommand}
\usepackage[notion, quotation, electronic]{knowledge}

\usepackage{import}

\hideLIPIcs

	\title{On the size of good-for-games Rabin automata and its link with the memory in Muller games}
	\titlerunning{On the size of good-for-games Rabin automata and the memory in Muller games} 

	\author{Antonio Casares}{LaBRI, Université de Bordeaux, France}{antonio.casares-santos@labri.fr}{https://orcid.org/0000-0002-6539-2020}{}
	\author{Thomas Colcombet}{CNRS, IRIF, Université Paris Cité, France}{thomas.colcombet@irif.fr}{https://orcid.org/0000-0001-6529-6963}{ANR Delta and Duall}
	\author{Karoliina Lehtinen}{CNRS, Aix-Marseille Université, Université de Toulon, LIS, France}{lehtinen@lis-lab.fr}{https://orcid.org/0000-0003-1171-8790}{}

	\authorrunning{A. Casares, T. Colcombet and K. Lehtinen} 
	\Copyright{Antonio Casares, Thomas Colcombet, and Karoliina Lehtinen} 
		
	\category{} 
	\relatedversion{} 
	\acknowledgements{We would like to thank Marthe Bonamy and Pierre Charbit for their help with graph theory.}
	\nolinenumbers 

\usepackage{tikz}
\usetikzlibrary{automata, snakes, positioning, arrows,arrows.meta,calc,decorations,decorations.markings,math,shapes.geometric}

\tikzset{
	>=stealth',
	-={stealth',ultra thick,scale=3} 
	node distance=1cm, 
	every state/.style={thick}, 
	initial text=$ $, 
}

\let\ab\allowbreak
\mathchardef\hyphen=45 

\definecolor{Green2}{HTML}{3EA514}
\definecolor{Red2}{HTML}{FF0400}
\definecolor{Orange2}{HTML}{E6670A}
\definecolor{Violet2}{HTML}{CE1ff9}
\definecolor{Green3}{HTML}{45A229}
\definecolor{Navy}{HTML}{2943A2}

\newrobustcmd\powerset{\mathcal P}

\DeclareMathAlphabet{\mathpzc}{OT1}{pzc}{m}{it}

\newrobustcmd{\NN}{\mathbb{N} }
\newrobustcmd{\ZZ}{\mathbb{Z}}
\newrobustcmd{\QQ}{\mathbb{Q}}
\newrobustcmd{\RR}{\mathbb{R}}
\newrobustcmd{\CC}{\mathbb{C}}
\newrobustcmd{\WW}{\mathbb{W}}

\newrobustcmd{\I}{\mathcal{I}}
\newrobustcmd{\F}{\mathcal{F}}
\newrobustcmd{\G}{\mathcal{G}}

\renewcommand{\L}{\mathcal{L}}
\newrobustcmd{\M}{\mathcal{M}}
\newrobustcmd{\Q}{\mathcal{Q}}
\newrobustcmd{\C}{\mathcal{C}}
\newrobustcmd{\A}{\mathcal{A}}
\newrobustcmd{\B}{\mathcal{B}}
\newrobustcmd{\Z}{\mathcal{Z}}
\newrobustcmd{\R}{\mathcal{R}}
\newrobustcmd{\T}{\mathcal{T}}
\newrobustcmd{\W}{\mathcal{W}}

\renewcommand{\P}{\mathcal{P}}

\newrobustcmd\muller{\mathcal F}
\newrobustcmd\rabin{\mathcal R}

\newrobustcmd\Lang[1]{\kl[\Lang{#1}]{\mathcal L}_{#1}}

\newrobustcmd{\kk}{\kappa}
\newrobustcmd{\uu}{\upsilon}
\newrobustcmd\dd{\kl[\dd]{\delta}}
\knowledge\dd{notion}

\renewcommand{\ss}{\sigma}

\newrobustcmd{\rr}{\rho}

\renewcommand{\aa}{\alpha}

\newrobustcmd{\bb}{\beta}

\newrobustcmd{\oo}{\omega}
\newrobustcmd{\pp}{\varphi}

\renewcommand{\gg}{\gamma}
\newrobustcmd{\ee}{\varepsilon}

\renewcommand{\SS}{\Sigma}
\newrobustcmd{\GG}{\Gamma}
\newrobustcmd{\DD}{\Delta}

\knowledgerenewmathcommand\nu{\cmdkl{\LaTeXnu}}
\knowledgerenewmathcommand\eta{\cmdkl{\LaTeXeta}}
\knowledgerenewmathcommand\chi{\cmdkl{\LaTeXchi}}

\knowledgenewrobustcmd\ZF{\mathcal{\cmdkl{\Z}_{\cmdkl{\F}}}}
\newrobustcmd{\Pplus}{"\P_+"}

\knowledgenewrobustcmd\memtree{\cmdkl{\mathrm{mem{\hyphen}tree}}}

\knowledgenewrobustcmd\zielonkatree{\cmdkl{\mathcal{Z}}}

\knowledgenewrobustcmd\RF{\cmdkl{\R_\F}}

\newrobustcmd{\prefix}{\kl[\minf]{\sqsubseteq}}
\newrobustcmd{\nprefix}{\sqsubset}

\newrobustcmd\nodes{N}
\knowledgenewrobustcmd\ancestor{\mathbin{\cmdkl{\sqsubseteq}}}

\newrobustcmd\memgen{\kl[\memgen]{\mathrm{mem}}}
\knowledge\memgen{notion}

\newrobustcmd\memZF{\kl[\memZF]{\mathrm{mem}_{\ZF}}}
\knowledge\memZF{notion}

\newrobustcmd{\mout}{"\mathit{Out}"}
\knowledge\mout{notion}
\newrobustcmd{\mIn}{\mathit{In}}
\newrobustcmd{\minf}{\kl[\minf]{\mathit{Inf}}}
\knowledge{\minf}{notion}
\newrobustcmd{\mlast}{\mathit{Last}}

\newrobustcmd{\mocc}{\kl{\mathit{Occ}}}
\newrobustcmd{\macc}{\mathit{Acc}}
\newrobustcmd{\mcolours}{\mathit{Colours}}

\newrobustcmd{\mletters}{\kl[\mletters]{\mathit{Letters}}}
\knowledge\mletters{notion}
\newrobustcmd{\maxpr}{\kl[\maxpr]{\mathit{Max\hyphen Priority}}}
\knowledge\maxpr{notion}

\newrobustcmd{\maccept}{\kl[\maccept]{\mathit{Accept}}}
\knowledge\maccept{notion}

\newrobustcmd{\moutput}{\kl[\moutput]{\mathit{Output}}}
\knowledge\moutput{notion}

\newrobustcmd{\msource}{\kl[\msource]{\mathit{Source}}}
\knowledge\msource{notion}

\newrobustcmd{\mtarget}{\kl[\mtarget]{\mathit{Target}}}
\knowledge\mtarget{notion}

\newrobustcmd{\mstates}{\kl[\mstates]{\mathit{States}}}
\knowledge\mstates{notion}

\newrobustcmd{\mnodes}{\kl[\mnodes]{\mathit{Nodes}}}
\knowledge\mnodes{notion}

\newrobustcmd{\mroundnodes}{\kl[\mroundnodes]{N_\bigcirc}}

\newrobustcmd{\msquarenodes}{\kl[\msquarenodes]{N_\Box}}

\newrobustcmd{\mleaves}{\kl[\mleaves]{\mathit{Leaves}}}
\knowledge\mleaves{notion}

\newrobustcmd{\msupport}{\kl[\msupport]{\mathit{Support}}}
\knowledge\msupport{notion}

\newrobustcmd{\mchildren}{\kl[\mchildren]{\mathit{Children}}}
\knowledge\mchildren{notion}

\newrobustcmd{\mnextn}{\kl[\mnextn]{\mathit{Next}_n}}
\knowledge\mnextn{notion}

\newrobustcmd{\mjump}{\kl[\mjump]{\mathit{Jump}}}
\knowledge\mjump{notion}

\newrobustcmd{\cmzielonka}{\mathtt{ChromaticMemory \hyphen ZT}}
\newrobustcmd{\memoryzielonka}{\mathtt{Memory \hyphen ZT}}
\newrobustcmd{\tinput}{\mathtt{Input}}
\newrobustcmd{\toutput}{\mathtt{Output}}
\newrobustcmd{\tproperty}{\mathtt{Property}}
\newrobustcmd{\Ptime}{\mathtt{P}}
\newrobustcmd{\NP}{\mathtt{NP}}
\newrobustcmd{\NPc}{\mathtt{NP}\hyphen\text{complete}}
\newrobustcmd{\CN}{\mathtt{Chromatic} \hyphen \mathtt{Number}}
\newrobustcmd{\nextmove}{\mathtt{next} \hyphen \mathtt{move}}

\knowledge{notion}
 | $\omega $-word
 | $\omega $-words
 | infinite words
 
\knowledge{notion}
| \P _+

\knowledge{notion}
| \lfloor q \rfloor

\knowledge{notion}
  | Muller condition
  | Muller conditions
  | Muller@condition
  | Muller

\knowledge{notion}
 | satisfies the Muller condition
 | satisfy@muller
 | satisfies@muller
 
\knowledge{notion}
 | language of the Muller condition
 | languages of Muller conditions
 | language of some Muller condition
 | \L _{\F _n}
 | \L _\F 
 | \Lang{\muller}
 
\knowledge{notion}
 | Muller language
 | Muller languages
 | Muller@language

\knowledge{notion}
  | Rabin condition
  | Rabin conditions
  | Rabin@condition
  | Rabin

\knowledge{notion}
 | satisfies the Rabin condition
 | satisfy@rabin
 | satisfies@rabin

\knowledge{notion}
 | language of the Rabin condition
 | languages of Rabin conditions
 | Rabin language
 | Rabin languages
 | Rabin@language
 | \Lang{\rabin}

\knowledge{notion}
 | Rabin pairs
 | Rabin pair

\knowledge{notion}
 | green in~$\gamma $
 | green
 | green in~$c$
 | green transition
\knowledge{notion}
 | red in~$\gamma $
 | red
 | red in~$c$
 
\knowledge{notion}
 | orange in~$c$
 | orange
 | orange transition
\knowledge{notion}
 | Streett condition
 | Streett@condition

\knowledge{notion}
 | parity condition
 | parity conditions
 | Parity conditions
 | parity@condition
 | parity

\knowledge{notion}
| satisfies the parity condition
| satisfy@parity
\knowledge{notion}
 | language of the parity condition
 | languages of parity conditions
 \knowledge{notion} 
 | generalised Büchi
 | generalised Büchi condition
 
 \knowledge{notion} 
 | generalised co-Büchi
 | generalised co-Büchi condition
 
\knowledge{notion} 
 | restriction of $\F $ to $C$
 | \F |_{\nu (n)}
 | restriction of $\F $ to
 | Muller condition restricted to $\nu (n_i)$
 
\knowledge{notion}
| non-deterministic automaton
| non-deterministic automata
| Non-deterministic automata
| automaton
| automata

\knowledge{notion}
 | $\WW $-automaton
 | Rabin automata
 | parity automata
 | Rabin automaton
 | parity automaton
 
\knowledge{notion} 
 | size@automaton
 
\knowledge{notion}
| deterministic Rabin automata
| deterministic Rabin automaton
| deterministic

\knowledge{notion}
| run of the automaton
| run over
| run

\knowledge{notion}
| state
| states
| States
\knowledge{notion}
| initial state
| initial states
| Initial states

\knowledge{notion}
| transition relation
| transition

\knowledge{notion}
| acceptance condition
| winning condition
| Acceptance condition

\knowledge{notion}
| accepting@run
| accepting run

\knowledge{notion}
| accepts@automaton
| accepted by the automaton
| accepted@automaton

\knowledge{notion}
| language accepted by the automaton
| accepting@language
| language accepted

\knowledge{notion}
| good-for-games automata
| good-for-games
| good-for-games Rabin-automata
| good-for-games Rabin-automaton
| GFG
| GFG Rabin-automata
| GFG automata
| GFG Rabin automaton
| GFG Rabin automata
\knowledge{notion}
| resolver

\knowledge{notion} 
| run induced by

\knowledge{notion}
| graph associated to

\knowledge{notion}
| cycle
| cycles

\knowledge{notion} 
| duplicated edges 
| duplicated transitions 

\knowledge{notion}
  | game
  | games
  | Games
  
\knowledge{notion}
 | transitions
 | moves 
 | move from~$x$
 
\knowledge{notion}  
  | Muller-game
  | $\WW $-game
  | $\WW '$-game
  | $L$-games
  | $K$-game
  | $L$-game
  | Muller games
  | $\muller $-game
  | $M$-games
  | $M$-game
  | Muller game

\knowledge{notion}
 | existential player
 | existential@player
 | Exist
 | the existential player
\knowledge{notion}
 | universal player
 | universal@player
 | Univ
\knowledge{notion}
 | strategy for the existential player
 | strategy@existential
 | strategies
 | strategy for the universal player
 | strategy@universal

\knowledge{notion}
  | \mathit {Out}

\knowledge{notion}
 | winning strategy for the existential player
 | winning@strategy
 | winning strategy for the universal player
 | winning strategy
\knowledge{notion}
 | existential player wins the game

\knowledge{notion}
 | universal player wins the game
 
\knowledge{notion}
 | winning@play

\knowledge{notion}
 | won by
 | wins
 | won by the existential player
 
\knowledge{notion}
 | partial play
\knowledge{notion}
 | consistent with
\knowledge{notion}
  | accepting loop
  | accepting@loop

\knowledge{notion}
  | rejecting loop
  | rejecting@loop

\knowledge{notion}
 | game
 | games

\knowledge{notion}
 | play

\knowledge{notion}
 | $\varepsilon $-free
 | $\varepsilon $-free games
 | $\varepsilon $-free game
 | $\ee $-free

\knowledge{notion}
 | memory structure for the game $\G $
 | memory structure 
 | memories
 | non-chromatic memory
 | memory
 | memory structures

\knowledge{notion}
 | chromatic memory
 | chromatic memories
 | chromatic memory structure
 | chromatic@memory
 | chromatic memory structures
 | Chromatic memories
 | chromatic
 
\knowledge{notion}
| arena-independent
| arena-independent memory
| Arena-independent memories
| arena-independent memories

\knowledge{notion} 
 | induces a strategy
 | induced
 | inducing
 | given by a memory
\knowledge{notion} 
 | size@memory
 
\knowledge{notion} 
  | update function
  
\knowledge{notion}
 | exist-positionality 
 | exist-positional
 | half-positional
 | half-positionality
 
\knowledge{notion}
 | positional strategies
 | positionality
 | positional
 | positional strategy
 | positional winning strategy

\knowledge{notion}
| \mathfrak{mem}_{\mathit{gen}}(\F)
| \mathfrak {mem}_{\mathit {gen}}(\F _n)
| memory requirements
| non-chromatic memory requirements
| general memory requirements
| Memory requirements
| general

\knowledge{notion}
| \mathfrak {rabin}(L_G)
| \mr (L)
| \mr(L_G) 
| \mr (L_\F )

\knowledge{notion}
| \mathfrak{mem}_{\mathit{chrom}}(\F)
| chromatic memory requirements

\knowledge{notion}
| \mathfrak{mem}_{\mathit{ind}}(\F)

\knowledge{notion}
 | tree
 | trees
\knowledge{notion}
 | subtree of~$T$ rooted at~$n$
 | subtrees of~$T$ rooted
 | subtree of~$\zielonkatree _\muller $ rooted at~$n$
 
\knowledge{notion}
 | \ancestor
 | ancestor relation
 | ancestors
 | ancestor 

\knowledge{notion}
 | descendant
 | descendants

\knowledge{notion}
 | ordered Zielonka tree
 | order
 | ordered

\knowledge{notion}
 | $A$-labelled tree
 | $\Pplus (\GG )$-labelled tree
\knowledge{notion}
 | the root
 | root
\knowledge{notion}
| height
\knowledge{notion}
 | leaves
 | leaf
\knowledge{notion}
 | parent
 | parents
\knowledge{notion}
 | node
 | nodes
\knowledge{notion}
 | child
 | children
\knowledge{notion}
 | \mroundnodes
 | round node
 | round nodes
 | Round nodes
 | round
\knowledge{notion}
 | square node
 | square nodes
 | Square nodes
 | square
 | \msquarenodes
\knowledge{notion}
 | Zielonka tree
 | Zielonka trees
 | \zielonkatree _\muller
 \knowledge{notion}
 | memory@zielonka
 | memories@zielonka

 \knowledge{notion}
| corresponding to

\knowledge{notion}
 | \G _{\F _n}

\knowledge{notion}
 | independent set

\knowledge{notion}
| colouring

\knowledge{notion}
| chromatic number
| \chi(G)
 
\knowledge{notion}
 | final $C$-Strongly Connected Component
 | $C$-FSCC
 | $C_i$-FSCC

\begin{document}

	\maketitle
	\begin{abstract}
		In this paper, we look at good-for-games Rabin automata that recognise a Muller language (a language that is entirely characterised by the set of letters that appear infinitely often in each word). We establish that minimal such automata are exactly of the same size as the minimal memory required for winning Muller games that have this language as their winning condition. We show how to effectively construct such minimal automata. Finally, we establish that these automata can be exponentially more succinct than equivalent deterministic ones, thus proving as a consequence that chromatic memory for winning a Muller game can be exponentially larger than unconstrained memory.
	\end{abstract}


\vspace{4mm}
	
	\section{Introduction}
	\label{section: Intro}

\textbf{Games.}
Games, as considered in this work, are played by two antagonistic players, called the existential and universal players, who move a token around finite edge-coloured directed graphs. When the token lands on a position belonging to one of the players, this player moves it along an outgoing edge onto a new position. At the end of the day, the players have constructed an infinite path, called a play, and the winner is determined based on some language~$\WW$ of winning infinite sequences of colours, called the "winning condition" (we call $\WW$-games the games which use the winning condition~$\WW$). Solving such games consists of deciding whether the existential player has a winning strategy, i.e. a way to guarantee, whatever the moves of the opponent are, that the play will end up in the winning condition.
Solving infinite duration games is at the crux of many algorithms used in verification, synthesis, and automata theory~\cite{Buchi77Games,Gurevich1982trees,PR89Synthesis, LMS20SynthesisLTL}. Difficulties in solving them are both theoretical and practical, and many questions pertaining to game resolution still remain unanswered. 

\textbf{Memory.}
Several parameters are relevant for solving a game: its size, of course, but also its "winning condition" and the complexity of winning "strategies". A measure of this complexity is the "memory" used by a strategy. The simplest strategies are those that use no memory ("positional strategies"): decisions depend exclusively on the current position, and not on the past of the game. A strategy uses a finite amount of memory if the information that we need to retain from the past can be summarized by a finite state machine that processes the sequence of moves played in the game. In this case, the amount of "memory" used by the strategy is the number of states of this machine. Given a "winning condition" $\WW$, a fundamental question is what is the minimal quantity $m$ such that if the existential player wins a $\WW$-game, there is a winning strategy using a memory of size $m$ (we call $m$ the "memory requirements" of $\WW$).
In addition to its size, a memory also has structure, which further elucidates the game dynamics. 
Understanding both the size and the structure of "memories" for $\WW$ is a crucial step to design algorithms for solving $\WW$-games.

\begin{quote}
	Question A: Give a structural description of the optimal memory in $\WW$-games.
\end{quote}

\textbf{Muller conditions.}
While there is a large zoo of winning conditions in the literature, here we are interested in $\omega$-regular ones (described by finite state automata over infinite words), and, in particular, so called "Muller conditions", for which the winner depends only on the colours that are seen infinitely often in the play. "Memory requirements" for "Muller conditions" have been studied in depth by Dziembowski, Jurdziński and Walukiewicz \cite{DJW1997memory}.
They provide a ``formula'' for computing the size of the minimal memory sufficient for winning in all games with a given  Muller winning condition, based on the "Zielonka tree"~\cite{Zielonka1998infinite}, which describes the structure of a "Muller condition".
The "Zielonka tree" has also been used to characterise the memory requirements of "Muller conditions" when randomised strategies are allowed~\cite{Horn09RandomFruits} and to provide minimal "parity automata" recognising a "Muller condition"~\cite{CCF21Optimal}. This fundamental structure is also at the heart of our contribution.

\textbf{Game reductions and good-for-gameness.} 
When confronted with a "$\WW$-game", a standard solution is to reduce it to a game with a larger underlying graph, but a simpler winning condition.
The typical way to do this  (but not the only one) is to perform the composition of the game with a suitable automaton with another acceptance condition~$\WW'$ that accepts the language~$\WW$.
The result is an "$\WW'$-game" which has as size the product of the size of the original game and the size of the automaton.
There is a subtlety here: not all automata can be used for this operation.
For a non-deterministic automaton, this is in general incorrect, while using a deterministic automaton is always correct.
So here, finding a minimal deterministic automaton for a given language improves the complexity of game resolution, and there is a large body of research in this direction (see \cite{CCF21Optimal,Kretinski2017IAR,Loding1999Optimal} for Muller conditions, \cite{AbuRadiKupferman19Minimizing,Casares2021Chromatic, Schewe10MinimisingNPComplete,Schewe20MinimisingGFG} for minimisation of automata, and
\cite{McNaughton1966Testing,Safra1988onthecomplexity,MullerSchupp95NewResults,Piterman2006fromNDBuchi,Schewe2009tighter,Michel1988Complementation,Loding1999Optimal,Colcombetz2009tight,LodingP19} for determinisation). 
However, some non-deterministic automata can also be used to perform this reduction. These are called "good-for-games automata" (GFG) \cite{HP06,Colcombet2009CostFunctions}. 

Some languages are known to be recognised by "good-for-games automata" that are exponentially more succinct than any equivalent deterministic automaton~\cite{KS15}, and several lines of research concerning good-for-games automata are under study (how to decide `good-for-gameness'~\cite{BK18,KS15,BKLS20,BL22}, how expressive is `good-for-gameness' for pushdown automata~\cite{LZ20,GJLZ21} what are good-for-games quantitative automata~\cite{BL21a}, etc).
However, one key question that has not yet been addressed concerning "good-for-games automata" is how to design techniques as general as possible for building them. To the best of our knowledge, the only existing result in this direction is a polynomial-time algorithm to minimise co-Büchi "GFG" automata~\cite{AbuRadiKupferman19Minimizing}.

\begin{quote}
	Question B: Provide general tools for constructing good-for-games automata.
\end{quote}


In this paper, in the context of "Muller conditions", we  relate these two lines of study, and in particular give partial answers to the general questions A and B. Indeed, we show that the "memory" needed to win in $L$-games for a "Muller language" $L$ coincides with the size of minimal "GFG Rabin automata" for $L$, and, in this sense, we give a structural description of the memory for "Muller games", thus giving a refined answer to question A in this case. We also provide an optimal way to construct these  minimal "good-for-games automata", thus answering question B in the context of "Muller conditions".

\paragraph*{Contributions.}
\begin{enumerate}
\item\label{item:contrib-GFG-memory} We show that for all $\omega$-regular languages~$L$, the "size@automata" (number of states) of a "good-for-games Rabin-automaton" for~$L$ is an upper bound on the memory
   that the existential player needs to implement winning strategies for "$L$-games".
   This inequality is straightforward, but had not been stated explicitly prior to this work.
\item We establish that when $L$ is a "Muller language", the following two quantities are equal: the least size of a "good-for-games Rabin-automaton" for~$L$
	and the least memory required for the existential player in all "$L$-games" in which she wins. Furthermore, we provide an efficient way to construct such a minimal automaton
	from the "Zielonka tree" of the condition \cite{Zielonka1998infinite}.
	This automaton can be seen, in a certain way, as a quotient of the minimal deterministic parity automaton for this language, as described in \cite{CCF21Optimal}. 

	Let us note that the least amount of memory needed to win a "Muller game" was  described precisely by Dziembowski,  Jurdziński and Walukiewicz \cite{DJW1997memory}. We show here that the optimal strategy described in  \cite{DJW1997memory} can be implemented in a "good-for-games Rabin-automaton". In combination with
	\Cref{item:contrib-GFG-memory}, this provides another proof of the upper bound in \cite{DJW1997memory}.
	\item Finally, we provide a family of "Muller languages" such that the smallest "GFG Rabin automata" recognising it are of linear size in the number of letters, while equivalent "deterministic Rabin automata"  grow exponentially. 
	Note that the least size of a "deterministic Rabin automaton" for a "Muller language" $L$ is known to coincide with the "chromatic memory" needed for winning "$L$-games"~\cite{Casares2021Chromatic} (i.e. a memory that is updated based only on the letters seen, independently of the position in the game).
	The question of equivalence between "chromatic memory" and "memory" was asked by Kopcyński \cite{Kopczynski2006Half,Kopczynski2008PhD}, and an arbitrary difference between these two notions was established only recently by Casares \cite{Casares2021Chromatic}. Our new result, which is incomparable, shows that the "chromatic memory" can grow exponentially in the size of the alphabet, even when the general "memory" remains linear.
\end{enumerate}
Together these three points develop techniques to solve Muller games in an optimal way by means of "good-for-games" "Rabin automata" reductions. 
The last point shows that an exponential gain can be achieved compared to using classical deterministic "Rabin automata".
Overall, our contribution supplements our understanding of "Muller languages" and  highlights the---so far unexplored---fundamental role of "GFG automata" in the equation. Indeed, up to now "GFG automata" had mainly been studied for their succinctness, expressivity or algorithmic properties. Here, we shed light on a novel dimension of this automata class. \\

\textbf{Related work.} There is vast amount of literature on the memory requirements of different games. The first results in this direction where the proofs of the "positionality" of "parity conditions" and "half-positionality" of "Rabin conditions"~\cite{EmersonJutla91Determinacy, Klarlund94Determinacy} and the finite-memory determinacy of "Muller games"~\cite{Gurevich1982trees}. The exact memory requirements of "Muller conditions" where characterised in~\cite{DJW1997memory}. In his PhD Thesis~\cite{Kopczynski2006Half, Kopczynski2008PhD}, Kopczyński characterises several classes of conditions that are "half-positional",  introduces the concept of "chromatic memories" (memories that are updated based only on colours seen) and provides an algorithm to decide the chromatic memory requirements of a winning condition. Conditions that are positional for both players over all graphs where characterised in~\cite{ColcombetN2006PositionalEdge} and those that are positional over finite graphs in~\cite{GimbertZielonka2005Memory}. More recently, these two results have been generalized to finite-memory conditions~\cite{BRORV20FiniteMemory, BRV22OmegaRegMemory}. The "memory requirements" have been proved to be different to the chromatic memory requirements in general~\cite{Casares2021Chromatic}, but conditions that are finite-memory determined are also chromatic-finite-memory determined~\cite{Kozachinskiy22Chromatic}.\\

\textbf{Structure of this document.}
In \Cref{section: Definitions}, we describe the classical definitions related to our work such as games, automata and good-for-gamesness.
In \Cref{section: GFGRabin-Memory}, we show why good-for-games Rabin automaton can be used as a memory structure for the existential player, the optimality of the construction for Muller conditions, and how to construct the least such automaton.
In \Cref{section:succinctness}, we establish that this construction can be exponentially more succinct than deterministic Rabin automata.
\Cref{section:conclusion} concludes the paper.

	\section{Definitions}
	\label{section: Definitions}


\textbf{Notations.} 
	\AP  $|A|$ denotes the cardinality of a set $A$, $\P(A)$ its power set and $\intro*\Pplus(A)=\P(A)\setminus\{\emptyset\}$.
	\AP For a finite non-empty alphabet $\SS$, we write $\SS^*$ and $\SS^\oo$ for the sets of finite and ""infinite words"" over $\SS$, respectively.
	The empty word is denoted by $\ee$.
	Given $w=w_0w_1w_2\dots\in\SS^\oo$, we denote $\intro*\minf(w)\subseteq\SS$ the set of letters that appear infinitely often in~$w$. We let $w[i..j]$ be the finite word $w_iw_{i+1}...w_j$ if $i\leq j$, and $\ee$ if $j<i$. 
	
	We extend maps $\ss: A \to B$ to $A^*$ and $A^\oo$ component-wise and we denote these extensions by $\ss$ whenever no confusion arises.
	\AP For a positive rational number $q\in\QQ$ we denote by $""\lfloor q \rfloor""$ the greatest integer $n\in \NN$ such that $n\leq q$.
	
\subsection{Games and their memory}\label{subsection:games}

 \AP  \textbf{Games.} We consider turn-based infinite duration games played between the existential and the universal player (referred to as  ""Exist"" and ""Univ"") over a directed graph. Formally, a $\GG$-coloured ""game"" is a tuple $\G=(V=V_E\uplus V_A, E, x_0, \WW)$, which consists of a set of vertices $V$ partitioned into Exist's positions $V_E$ and Univ's ones, $V_A$; a set of ""transitions"" (also called edges or moves) $E\subseteq V \times (\GG\cup\{\ee\}) \times V$; an initial vertex $x_0\in V$ and a subset $\WW\subseteq \GG^\oo$ of winning sequences. 
 We make the assumptions that there is at least one move from every position and that no cycle is labelled exclusively by $\ee$.
 We will denote by $\gg: E \to \GG\cup \{\ee\}$ the function that assigns to each edge its colour.
 \AP We write $\intro*\mout(x)$ for the set of outgoing moves from $x$, that is, $\mout(x)=\{e\in E \: : \: e=(x,c,x') \text{ for some } x'\in V, c\in \GG\cup\{\ee\}\}$. If $\GG$ is a game using the winning condition $\WW$ we call it a ""$\WW$-game"".

\AP Each player moves a pebble along an outgoing edge whenever it lands on a position belonging to that player, forming an infinite path $\pi\in E^\oo$ starting in $x_0$ called a ""play"".

We denote $\gg(\pi)$ sequence of colours labelling $\pi$ omitting the $\ee$ labels (we remark that $\gg(\pi)\in \GG^\oo$, since there are no cycles entirely labelled by $\ee$).
The play is ""winning@@play"" for the "existential player" if $\gg(\pi)\in \WW$. 
A ""partial play"" is a finite path $\pi\in E^*$ in $\G$ starting in $x_0$.
A ""strategy for the existential player"" is a function $\ss\colon E^*\to E$ such that if a "partial play" $\pi$ ends in a position $x\in V_E$, then $\ss(\pi)\in \mout(x)$. 
We say that a "play" $\pi$ is ""consistent with"" the strategy $\ss$ if for every "partial play" $\pi'$ that is a prefix of $\pi$ ending in a position controlled by "Exist", the next edge in $\pi$ is $\ss(\pi')$.
The "strategy@@existential" $\ss$ is ""winning@@strategy"" if every "play" "consistent with" $\ss$ is "winning@@play" for the existential player.
We say that the game $\G$ is ""won by"" "the existential player" if
that player has a "winning strategy" in $\G$.
\AP A strategy is ""positional"" if it can be represented by a function $\ss\colon V_E \to E$ (that is, the choice of the next transition only depends on the current position, and not on the history of the path). \\

\noindent\textbf{Winning conditions.}
	We fix an alphabet~$\Gamma$.
	\begin{description}
		\itemAP[Muller.] A  ""Muller condition"" over the alphabet~$\Gamma$ is given by a family $\muller\subseteq \Pplus(\GG)$. A word $w\in \Gamma^\oo$ ""satisfies the Muller condition""~$\muller$ if $\minf(w)\in\muller$.
		\AP The ""language of the Muller condition"" $\Lang\muller\subseteq\Gamma^\omega$ contains the "$\omega$-words" that "satisfy@@muller"~$\muller$.

		\itemAP[Rabin.] A  ""Rabin condition"" over the alphabet~$\Gamma$ is represented by a family of ""Rabin pairs"" $R=\ab\{(G_1,R_1),\dots,\ab(G_r,R_r)\}$, where $G_i,R_i\subseteq \Gamma$ and $G_i\cap R_i= \emptyset$. 
	The "Rabin pair"~$j$ is said to be ""green in~$c$"" if $c\in G_j$, to be ""red in~$c$"" if~$c\in R_j$, or to be ""orange in~$c$""  if none of the previous occur.
		 A word $w\in \Gamma^\omega$ ""satisfies the Rabin condition""~$R$ if $\minf(w) \cap G_j \neq \emptyset$ and $ \minf(w) \cap R_j = \emptyset $ for some index $j\in \{1,\dots,r\}$. Said differently, there is a "Rabin pair"~$j$ which is "red"  in finitely many letters from~$w$, and "green" for infinitely many letters of~$w$.
			\AP The ""language of the Rabin condition"" $\Lang\rabin\subseteq\Gamma^\omega$ contains the "$\omega$-words" that "satisfy@@rabin"~$\rabin$.
		
	\itemAP[Parity.] To define a ""parity condition"" we suppose that $\GG\subseteq \NN$. A word $w\in \GG^\oo$  ""satisfies the parity condition"" if the maximum in $\minf(w)$ is even.
	\AP The ""language of the parity condition"" contains the "$\omega$-words" that "satisfy@@parity" it.
	\end{description}
	
	We say that a language $L\subseteq \GG^\oo$ is a ""Muller language"" if it is the "language of some Muller condition"~$\muller$. Equivalently, $L$ is a Muller language if it can be described as a boolean combination of atomic propositions of the form ``the letter `$a$' appears infinitely often'' and their negations.
	Note that  "languages of Rabin conditions" are "languages of Muller conditions", and that "languages of parity conditions" are "languages of Rabin conditions", but the converses do not hold.
	\AP Given a Muller condition $\F\subseteq \Pplus(\GG)$ and a subset $C\subseteq \GG$, we define the ""restriction of $\F$ to $C$"" as the Muller condition over $C$ given by $\F|_C=\{A\subseteq C \: : \: A\in \F\}$.\\

\noindent\textbf{Memory structures.}
	\AP A ""memory structure for the game $\G$"" of "moves"~$E$ is a tuple $\M=(M, m_0, \mu,\sigma)$ where
	   $M$ is a \emph{set of memory states},
	   $m_0\in M$ is an \emph{initial memory state},
	   $\mu: M \times E \to M$ is an ""update function"", and
	   $\sigma: M\times V_E \to E$ maps each position~$x$ owned by the "existential player" to a "move from~$x$". 
	The ""size@@memory"" of~$\M$ is the cardinal of $M$.
	We extend the function $\mu$ to paths by induction: $\mu(m,\varepsilon)=m$,
	and~$\mu(m, \pi e)=\mu(\mu(m,\pi),e)$ for a path $\pi \cdot e \in E^*$. The memory structure $\M$ ""induces a strategy"" $\ss_\M\colon E^*\to V_E$ given by $\ss_\M(\pi)=\sigma(\mu(m_0,\pi), \mlast(\pi))$, where $\mlast(\pi)$ is the last position of the "partial play" $\pi$.
	
	\AP We say that $\M$ is a \AP""chromatic memory structure"" if there is a function $\mu_c\colon M\times (\GG\cup\{\ee\}) \to M$ such that $\mu_c(m,\ee)=m$ for every $m\in M$ and $\mu(m,e) = \mu_c(m,\gg(e))$ for all edges $e\in E$. 

	 \AP The ""memory requirements"" of a winning condition $\WW$ are defined as the least integer $n$ such that if $\G$ is an "$\WW$-game" "won by the existential player", then she has a "winning strategy" "given by a memory" of size at most $n$. We denote this quantity by $\intro*\memgen(\WW)$.

	 \AP We say that a "winning condition" $\WW$ is ""exist-positional"" (also called \emph{half-positional})
	 if $\memgen(\WW)=1$. Equivalently, $\WW$ is exist-positional if "Exist" has a winning positional strategy whenever "Exist" has a winning strategy at all.

\subsection{Automata and good-for-gameness}\label{subsection:automata}

	\noindent\textbf{Automata.} A ""non-deterministic automaton"" (or simply an \reintro{automaton}) $\A=(Q, \Sigma, Q_0, \DD, \GG, \WW)$ consists of a finite set of ""states"" $Q$, an input alphabet $\Sigma$, a non-empty set of ""initial states"" $Q_0\subseteq Q$, a ""transition relation"" $\Delta \subseteq Q\times \Sigma \times \GG \times Q$ and an ""acceptance condition"" $\WW\subseteq \GG^\omega$.
	We will write $\intro*\dd\colon Q \times \SS \to \P(Q)$ for the function $\dd(q,a)=\{q'\in Q \: : \: (q,a,c,q')\in \DD \text{ for some } c\in \GG\}$.
	The ""size@@automaton"" of the "automaton", $|\A|$, is the number of its states.
	\AP A ""run of the automaton"" over a word~$a_0a_1a_2\dots\in\Sigma^\omega$ is a sequence of transitions of the form:
	\begin{align*}
		\rho = (q_0,a_0,c_0,q_1)(q_1,a_1,c_1,q_2)\dots \in\Delta^\omega, \text{ such that } q_0\in Q_0 \text{ is an "initial state".}
	\end{align*}
	
	\AP A run is ""accepting@@run"" if ~$c_0c_1c_2\dots\in\WW$.
	If an "accepting run" over a word~$w\in \SS^\oo$ exists, the "automaton" ""accepts@@automaton""~$w$.
	The set of accepted words is the ""language accepted by the automaton"", written $L(\A)$.
	\AP The "automaton" is ""deterministic"" if $Q_0$ is a singleton and $\Delta$ is such that for all states~$q$ and letter~$a\in\Sigma$, there exists exactly one transition of the form $(q,a,c,q')$.
	In this case, for all words~$w\in\Sigma^\omega$ there exists one and exactly one "run" of the automaton over~$w$.
	
	An automaton using an "acceptance condition" $\WW$ (resp. an acceptance condition of type $X\in \{\text{"Muller", "Rabin", "parity"}\}$) is called an ""$\WW$-automaton"" (resp. $X$-automaton).\\
	

	\noindent\textbf{Good-for-gameness.} The automaton $\A$ is ""good-for-games"" (GFG) if there is a \textit{""resolver""} for it, consisting of a choice of an initial state $r_0\in Q_0$ and a function $r:\Sigma^*\times \Sigma \rightarrow \Delta$ such that for all words $w\in L(\A)$, the run $t_0t_1...\in \DD^\oo$, called the ""run induced by"" $r$ and defined by $t_i=r(w[0..i-1],w[i])$, starts in $r_0$ and is an "accepting run" over $w$. In other words, $r$ should be able to construct an accepting run in $\A$ letter-by-letter with only the knowledge of the word so far, for all words in $L(\A)$.

\subsection{The Zielonka tree of a Muller condition}\label{subsection:Zielonka}

	\AP A ""tree""~$T=(N,\ancestor)$ is a nonempty finite set of ""nodes""~$\nodes$ equipped with an order relation $\ancestor$ called the ""ancestor relation"" ($x$ is an ancestor of $y$ if $x\ancestor y$), such that
		(1) there is a minimal "node" for~$\ancestor$, called ""the root"", and
		(2) the "ancestors" of an element are totally ordered by~$\ancestor$.
	\AP The converse relation is the ""descendant"" relation.
	\AP Maximal nodes are called ""leaves"", and the set of leaves of $T$ is denoted by $\intro*\mleaves(T)$. 
	\AP Given a "node"~$n$ of a tree~$T$, the ""subtree of~$T$ rooted at~$n$"" is the tree~$T$ restricted to the nodes that have~$n$ as ancestor.
	\AP A node $n'$ is a ""child"" of $n$ if it is a minimal strict "descendant" of it. The set of children of $n$ is written $\intro*\mchildren_T(n)$.
	\AP The ""height"" of a tree $T$ is the maximal length of a chain for the "ancestor" relation. 
	\AP An ""$A$-labelled tree"" is a tree $T$ together with a labelling function $\nu \colon N \to A$.
	\begin{definition}[\cite{Zielonka1998infinite}]
		Let~$\muller\subseteq\Pplus(\GG)$ be a "Muller condition". A ""Zielonka tree"" for~$\muller$, denoted $\intro*\zielonkatree_\muller = (N, \ancestor, \intro*\nu :N \to \Pplus(\GG))$ is a "$\Pplus(\GG)$-labelled tree" with nodes 
		 partitioned into ""round nodes"" and ""square nodes"", $N= \mroundnodes \uplus \msquarenodes$ such that:
		\begin{itemize}
		\item The "root" is labelled $\GG$.
		\item If a "node" is labelled~$X\subseteq \GG$, with~$X\in\muller$, then it is a "round node", and its children are labelled exactly with the maximal subsets~$Y\subseteq X$ such that~$Y\not\in\muller$.
		\item If a "node" is labelled~$X\subseteq \GG$, with~$X\not\in\muller$, then it is a "square node", and its children are labelled exactly with the maximal subsets~$Y\subseteq X$ such that~$Y\in\muller$.
		\end{itemize}
	\end{definition}

	We remark that if $n$ is a node of $\zielonkatree_\muller$, then the "subtree of~$\zielonkatree_\muller$ rooted at~$n$" is a Zielonka tree for $"\F|_{\nu(n)}"$, (the "restriction of $\F$ to" the label of $n$).
	
	
	We equip trees with an order in order to navigate in them. An \AP""ordered Zielonka tree"" is a "Zielonka tree" for which the set of children of each node is equipped with a total order(``from left to right''). The function $\intro*\mnextn$ maps each child of~$n$ to its successor for this order, in a cyclic way.
	For each node $n\in N$ and each "leaf" $l$ below $n$ we define the set $\intro*\mjump_n(l)$ containing a leaf $l'$ if there are two children~$n_1, n_2$ of~$n$ such that 
	$n_1\ancestor l$, $n_2\ancestor l'$ and $n_2=\mnextn(n_1)$ (we remark that $n_1=n_2$ if $n$ has only one child). For $n=l$ we define $\mjump_n(l)=\{l\}$. That is, $l'\in \mjump_n(l)$ if we can reach $l'$ by the following procedure: we start at $l$, we go up the tree until finding the node $n$, we change to the next branch below $n$ (in a cyclic way) and we re-descend to $l'$.	
	From now on, we will suppose that all "Zielonka trees" are "ordered", without explicitly mentionning it.

\begin{example}\label{example: Zielonka-tree}
	We will use the following "Muller condition" as a running example throughout the paper.
	Let $\GG= \{a,b,c\}$ and let $\F$ be the Muller condition defined by:
	\[ \F = \{\{a,b\}, \{a,c\}, \{b\}\}. \]
	In Figure~\ref{Fig: ZielonkaTree-example} we show the "Zielonka tree" for $\F$.	
	We use Greek letters to name the nodes of the tree, $N$. 
	We have that $\mjump_\aa(\delta)=\{\ee, \zeta\}$ and $\mjump_\gg(\zeta)=\{\ee\}$.
	The numbering of the branches will be used in Section~\ref{subsection:GFGforMuller}.
	
	\begin{figure}[ht]
		\centering
		\begin{tikzpicture}[square/.style={regular polygon,regular polygon sides=4}, align=center,node distance=2cm,inner sep=3pt]
		
		\node at (0,2.4) [draw, rectangle, minimum height=0.8cm, minimum width = 1.3cm] (R) {$a,b,c$};
		
		\node at (-1,1.2) [draw, ellipse,minimum height=0.8cm, minimum width = 1.3cm, scale=1] (0) {$a,b$};
		\node at (1,1.2) [draw, ellipse, minimum height=0.8cm, minimum width = 1.3cm, scale=1] (1) {$a,c$};
		
		\node at (-1,0) [draw, square,minimum width=0.9cm, scale=1] (00) {$a$};
		\node at (0.5,0) [draw, square,minimum width=0.9cm, scale=1] (10) {$a$};
		\node at (1.5,0) [draw, square,minimum width=0.9cm, scale=1] (11) {$c$};
		
		\node at (0.9,2.4)  (alpha) {$\textcolor{Violet2}{  \aa }$};
		\node at (-1.9,1.2)  (beta) {$\textcolor{Violet2}{ \bb }$};
		\node at (1.9,1.2)  (gamma) {$\textcolor{Violet2}{ \gg } $};
		\node at (-1.6,0)  (delta) {$\textcolor{Violet2}{ \delta }$};
		\node at (0,0)  (eps) {$\textcolor{Violet2}{ \ee } $};
		\node at (2,0)  (zeta) {$\textcolor{Violet2}{ \zeta }$};
		
		\node at (-1,-0.6)  (B1) {\textcolor{Navy}{  \textbf{1} }};
		\node at (0.5,-0.6) (B2) {\textcolor{Green3}{  \textbf{1} }};
		\node at (1.5,-0.6) (B3) {\textcolor{black}{  \textbf{2} }};
		
		\draw   
		(R) edge (0)
		(R) edge (1)
		
		(0) edge (00)
		
		(1) edge (10)
		(1) edge (11);
		
		\end{tikzpicture}
		\caption{Zielonka tree $\zielonkatree_\F$ for 
			$\F=\{ \{a,b\}, \{a,c\}, \{b\}\}$.}
		\label{Fig: ZielonkaTree-example}
	\end{figure}
\end{example}

	
\begin{definition}[\cite{DJW1997memory}]
	Let $T$ be a tree with "nodes" partitioned into "round" and "square nodes", its ""memory for the existential player@memory@zielonka"" (\reintro(zielonka){memory} for short), denoted $\intro*\memtree(T)$, is defined inductively as:
	\begin{itemize}
	\item $1$ if $T$ has exactly one "node".
	\item The sum of the "memories@@zielonka" of the "subtrees of~$T$ rooted" at the "children" of the "root", if the "root" is "round".
	\item The maximum of the "memories@@zielonka" of the "subtrees of~$T$ rooted" at the "children" of the "root", if the "root" is "square".
	\end{itemize}
%
\end{definition}

For instance, for the "Zielonka tree" from Example~\ref{example: Zielonka-tree}, $\memtree(\zielonkatree_\F)=2$.\\

The key result justifying the introduction of this notion is that it characterises precisely the quantity of memory required for winning an~"$\muller$-game", as shown by the next proposition.
\begin{proposition}[\cite{DJW1997memory}]\label{prop: optimalMemoryDJW}
	For all "Muller conditions" $\muller$, $\memgen(\muller) = \memtree(\zielonkatree_\muller)$.
\end{proposition}

	\section{GFG Rabin automata correspond to memory structures for Muller games}
	\label{section: GFGRabin-Memory}

In this section, we prove the following result:

\begin{theorem}\label{Th_ GFGRabin=Memory}
	Let $L$ be a "Muller language". The "memory requirements" for $L$ coincide with the size of a minimal "GFG" "Rabin automaton" recognising $L$. 
\end{theorem}

In \Cref{subsection:GFGtomemory} we show the first direction: the size of a "GFG" "Rabin automaton" for a "Muller language"~$L$ is always an upper bound on the "memory" required by the "existential player" on "$L$-games".
In \Cref{subsection:GFGforMuller}, we show how to construct a "GFG" "Rabin automaton" from the "Zielonka tree" of a "Muller condition" of size~$\memtree(\zielonkatree_L)=\memgen(L)$, completing the equivalence. Moreover, we can build this minimal "GFG" Rabin automaton in polynomial time given the "Zielonka tree" of the "Muller condition".

\subsection{A GFG Rabin automaton induces a memory structure for any game}\label{subsection:GFGtomemory}

In this section, we establish \Cref{corollary:automata>=memory}
which states that the "size@automaton" of a "GFG" "Rabin automaton"~$\A$ accepting a "Muller language"~$L$
is an upper bound on the memory required for winning all "$L$-games". Concretely, given an "$L$-game" "won by" the "existential player",
we are able to construct a "memory structure"  "inducing" a "winning strategy" based on $\A$.
The argument is standard: we construct the product game of $\A$ and the "$L$-game", which is a Rabin game in which the "existential player" enjoys a "positional winning strategy"; then we use the $\A$-component of this product as a "memory structure" for a strategy in the original "$L$-game". 

\begin{lemma}\label{lemma: GFGinducesMemory}
	Let~$\A=(Q,\Sigma,Q_0,\Delta,\Gamma, \WW)$ be a  "GFG" "$\WW$-automaton"
	recognising a language
	$L\subseteq\Sigma^\omega$, with~$\WW$ "exist-positional".
	Then if $\G$ is an "$L$-game" won by "Exist", she can win it using a strategy given by a "memory structure" $\M=(Q,r_0,\mu,\ss)$.
\end{lemma}
\begin{proof}
	Let $\G=(V=V_E\uplus V_A,E,x_0,L)$ and $\A=(Q,\Sigma,r_0,\Delta,\Gamma, \WW)$, where $r_0\in Q_0$ is the initial state chosen by some "resolver".
	We consider the product "$\WW$-game" $\G'$ in which:
	\begin{itemize}
		\item Positions are elements in $V'=(V\times Q)\cup (V\times \GG \times \P(Q))$. The initial position is $(x_0,r_0)$.
		\item "Exist"'s positions are $V_E'=(V_E\times Q) \cup (V\times \GG \times \P(Q))$.
		\item There is an $\varepsilon$-coloured edge from $(x,q)$ to $(x',c,\dd(q,c))$  if $(x,c,x')\in E$, $c\neq \varepsilon$, and to $(x',q)$ if $(x,\ee,x')\in E$.
		\item There is a $c$-labelled edge from $(x,c,S)$ to $(x,q')$ for all $q'\in S\subseteq Q$.
		\item The "winning condition" is $\WW$.
	\end{itemize}
	
	In short, in this game the players still negotiate a play in $G$, but in addition, the "existential player" must simultaneously build an accepting run on the labelling of this play in $\A$.
	
	If $\A$ is GFG then whenever the existential player wins in $\G$, she also wins in $\G'$~\cite[Theorem 3]{HP06} by playing a "winning strategy" on the $\G$ component of $\G'$ and using the "resolver" for $\A$ to choose the successor state in the $\A$ component. 
	
	This strategy is not necessarily positional. However, since $\WW$ is "exist-positional", the "existential player" also has a positional strategy $s:V_E'\to E'$. We can now build a memory structure $\M_s=(Q,r_0,\mu,\sigma)$ that projects the strategy $s$ onto $\G$ (and updates itself with the projection of $s$ onto $\A$):
	\begin{itemize}
		\item $\mu(q,e)= q$ if $e$ is an $\varepsilon$-coloured move of $\G$ and
		\item $\mu(q,(x,c,x'))=q'$ where 
		$q'=s(x',c, \dd(q,c))$ otherwise;
		\item $\sigma(q,x)=s(x,q)$.
	\end{itemize} 
	
	Since $s$ is a winning strategy in $\G'$, its projection onto $\G$ is a strategy that only agrees with plays with labels in $L(\A)$, that is, winning plays. 
\end{proof}

\begin{remark}
	Note that there is a slight subtlety here: the "resolver", which induces a "winning strategy" in the product game, does not need to be positional. In fact, it might require exponential memory~\cite[Theorem 1]{KS15}. Yet, for each "$L$-game", a memory based on $\A$ suffices.
\end{remark}

\begin{lemma}[\cite{Klarlund94Determinacy,Zielonka1998infinite}]
	"Rabin conditions" are "exist-positional".
\end{lemma}

As a direct consequence we obtain the following Proposition.
\begin{proposition}\label{prop: GFGRabin->Memory}
	Let $L$ be a "Muller language" "accepted@@automaton" by a "GFG" "Rabin automaton"  $\A$. 
	Then, in every "$L$-game" $\G$ won by the existential player, she can win using a strategy given by a "memory structure" of size $|\A|$.
\end{proposition}
\begin{corollary}\label{corollary:automata>=memory}
	If $\A$ is a "GFG" "Rabin automaton" "accepting@@language" a "Muller language" $L$, then $\memgen(L)\leqslant|\A|$.
\end{corollary}

\subsection{An optimal construction of a GFG Rabin automaton}\label{subsection:GFGforMuller}

So far, we have seen that given a "GFG Rabin automaton", it can serve as a "memory structure" for the games of which it "accepts@automaton" the "winning condition". In this section we do the converse: we build a minimal "GFG Rabin automaton" for a "Muller language" $L$ of the same size as the minimal memory required to win in $L$-games. 
\subsubsection{The construction}

\begin{proposition}\label{prop: Construction_GFGRabin}
	Let $\F\subseteq \Pplus(\GG)$ be a "Muller condition". There exists a "GFG Rabin automaton" recognising $"\L_\F"$ of size $\memgen(\muller)$.
\end{proposition}

To prove Propositon~\ref{prop: Construction_GFGRabin}, we build a "GFG" "Rabin" "automaton" $\RF=(Q, \GG, q_0, \DD, \GG', R)$
for $"\L_\F"$ based on the "Zielonka tree" $\zielonkatree_\muller$, as illustrated in~\cref{Fig_GFGRabin}. 
We use a mapping from the leaves of $"\zielonkatree_\muller"$ to the "states" of the "automaton" that guarantees that two leaves of which the last common ancestor is a "round node" cannot map to the same state. The number of states required to satisfy this condition (\ref{Eq_property-star} below) coincides with $\memtree(\zielonkatree_\muller)$. Then, for each leaf of $"\zielonkatree_\muller"$ and letter $c\in \Gamma$, we identify its last ancestor $n$ in $\zielonkatree_\muller$ containing $c$, and, using the $\mjump_n$ function (defined in~\cref{subsection:Zielonka}), pick a leaf below the next child of $n$. We add a $c$-transition with label $n$ between the states mapped to from these leaves. This way, we can identify a "run" in the automaton with a promenade through the nodes of the "Zielonka tree". If during this promenade a unique minimal node (for $\ancestor$) is visited infinitely often, it is not difficult to see that the sequence of input colours belongs to $\F$ if and only if the label of this minimal node is an accepting set (it is a "round node"). We devise a "Rabin condition" over the set of nodes of the "Zielonka tree" accepting exactly these sequences of nodes.	 


We now describe the construction of the "automaton" $\AP\intro*\RF=(Q, \GG, q_0, \DD, N, R)$ formally, starting from the "Zielonka tree" $\zielonkatree_\muller = (N, \ancestor, \nu:N \to \Pplus(\GG))$, and then proceed to prove its correctness.\\

\noindent\textbf{States.}
\AP First, we set $Q=\{1,2,\dots, \memtree(\zielonkatree_\muller)\}$ and we label the "leaves" of $\zielonkatree_\muller$ by a mapping $""\eta"": \mleaves(\zielonkatree_\muller) \to \{1,2,\dots, \memtree(\zielonkatree_\muller)\}$ verifying the property:
\begin{align}\label{Eq_property-star}
\nonumber&\text{If } n \in \zielonkatree_\muller \text{ is a "round node" with children } n_1 \neq n_2 \text{, for any pair}\\
\tag{$\star$}&\text{of leaves } l_1 \text{ and } l_2 \text{ below } n_1 \text{ and } n_2 \text{, respectively, } \kl{\eta}(l_1)\neq \kl{\eta}(l_2).
\end{align}

\begin{lemma}
	For every "Zielonka tree" $\zielonkatree_\muller$ there is a mapping verifying Property~\ref{Eq_property-star} of the form $\kl{\eta}: \mleaves(\zielonkatree_\muller) \to \{1,2,\dots, \memtree(\zielonkatree_\muller)\}$ .
\end{lemma}

\begin{proof}
	We prove it by induction in the "height" of $\zielonkatree_\muller$. Let $n_1,\dots, n_k$ be the "children" of the root of $\zielonkatree_\muller$. We write $\F_i$ to denote the "Muller condition restricted to $\nu(n_i)$" and $\kl{\eta}_i: \mleaves(\zielonkatree_{\muller_i}) \to \{1,2,\dots, \memtree(\zielonkatree_{\muller_i})\}$ be a labelling verifying Property~\ref{Eq_property-star}, for $1\leq i \leq k$.  We distinguish two cases according to the shape of the root.
	
	If the root of $\zielonkatree_\muller$ is a "square node" ($\GG\notin \F$), then the mapping $\kl{\eta}(l)=\kl{\eta}_i(l)$ if the leaf $l$ belongs to the subtree $\zielonkatree_{\muller_i}$ verifies Property~\ref{Eq_property-star}.
	
	If the root of $\zielonkatree_\muller$ is a "round node" ($\GG\in \F$), then $\memtree(\zielonkatree_\muller)=\sum_{i=1}^k \memtree(\zielonkatree_{\muller_i})$, and we can partition $\{1,2,\dots, \memtree(\zielonkatree_\muller)\}$ into disjoint sets $C_1,\dots, C_k$ of size $|C_i|=\memtree(\zielonkatree_{\muller_i})$. We write $\ss_i$ for a bijection from $\{1,\dots, \memtree(\zielonkatree_{\muller_i})\}$ to $C_i$. Then, the mapping
	$\kl{\eta}(l)=\ss_i(\kl{\eta}_i(l)), $ if the leaf $ l$ belongs to the subtree $\zielonkatree_{\muller_i}$ verifies Property~\ref{Eq_property-star}.
\end{proof}

We suppose that the image of the leftmost leaf under $\kl{\eta}$ is $1$ and choose the initial state $q_0$ to be $1$.
In Example~\ref{example: Zielonka-tree}, the labelling $\kl{\eta}(\delta)=\kl{\eta}(\ee)=1$, $\eta(\zeta)=2$ verifies Property~\ref{Eq_property-star}.\\

\noindent\textbf{Transitions.} 
For each leaf $l\in \mleaves(\zielonkatree_\muller)$ and each letter $c\in \GG$, we define a $c$-"transition" from $\eta(l)$, with an output label from $\GG'=N$, as follows: let $n$ be the maximal ancestor of $l$ that contains the letter $c$ in its label and let $l'$ be the leftmost leaf in $\mjump_n(l)$\footnote{We could add all transitions $\{(\eta(l), n, \eta(l')) \: : \: l'\in \mjump_n(l)\}$ to $\DD$. However, a "resolver" for the "GFG" automaton just needs to make use of one of these transitions, so in order to simplify the automaton we make an arbitrary choice (the leftmost leaf).}. Then, $(\eta(l),c, n, \eta(l'))\in \DD$.
That is, if we are in a state $\eta(l)$, when we read the letter $c\in\GG$ we can choose to go up in the Zielonka tree from $l$ until visiting a node $n$ with $c$ in its label. We produce the letter $n$ as output, then we change to the next child of $n$ (in a cyclic way) and we descend to the leftmost leaf below it. The destination is the $\eta$-label of this leaf\footnote{We remark that $l'$ is the target of the transition of the Zielonka tree parity automaton from $l$ reading letter `$c$'~\cite{CCF21Optimal}. See also \Cref{subsubsection: relationWithZielonkaParity}.}. 


Following the above definition, we obtain a mapping from transitions in the automaton to $\mleaves(\zielonkatree_\F)\times \GG \times N \times \mleaves(\zielonkatree_\F)$. We say that $l$ and  $l'$ are the "leaves" ""corresponding to"" the transition $(q,a,n,q')$  if this transition is sent to $(l,a,n,l')$ by this mapping. The node $n$ produced as output  is the last common ancestor of $l$ and $l'$.
The automaton obtained in this way might present multiple transitions labelled by the same input letter between two states. We will show in Proposition~\ref{prop: simplification_Rabin} that duplicated transitions can be removed.\\

\noindent\textbf{"Acceptance condition".}
We define a "Rabin condition" over the alphabet $\GG'=N$, that is the set of nodes of the "Zielonka tree". We define a "Rabin pair" for each "round node" of $\zielonkatree_\muller$ (that is, nodes whose label is an accepting set of letters for $\F$): $R=\{(G_n,R_n)\}_{n\in\mroundnodes}$. Let $n$ be a "round node" and $n'$ be a general node of $\zielonkatree_\muller$:
\begin{equation*}
\begin{cases}
n' \in G_n & \text{ if } n'=n,\\
n' \in R_n & \text{ if }  n'\neq n \text{ and } n \text{ is not an "ancestor" of } n'.
\end{cases}
\end{equation*}

That is, for the letter $n'$, the "Rabin pairs" corresponding to "round" ancestors of  $n'$ are not affected by it (they are ``"orange" in $n'$''). If this node $n'$ is "round", then it belongs to $G_{n'}$ (this pair is ``"green" in $n'$''). For any other node $n\in \mroundnodes$, we have $n' \in R_n$ (the pair is ``"red" in $n'$''). 


\begin{remark}
	The construction presented depends on the "order" of the nodes of the "Zielonka tree". However, the size of the resulting automaton is independent of this order.
\end{remark}

\begin{example}\label{example: GFG-Rabin}
	Let $\F = \{\{a,b\}, \{a,c\}, \{b\}\}$ be the "Muller condition" from Example~\ref{example: Zielonka-tree}. The labelling of the leaves of the "Zielonka tree" given by $\eta(\delta)=\eta(\ee)=1$, $\eta(\zeta)=2$ verifies Property~\ref{Eq_property-star}. Figure~\ref{Fig_GFGRabin} shows the "GFG" Rabin automaton obtained by following this procedure.
	
	\begin{figure}[ht]
		\centering		
		\begin{tikzpicture}[square/.style={regular polygon,regular polygon sides=4}, align=center,node distance=2cm,inner sep=2pt]
		
		\node at (0,2) [state] (1) {$1$};
		\node at (3,2) [state] (2) {$2$};
		
		\path[->] 
		(1)  edge [in=160,out=200,loop, color=Navy] node[left] {$a:\textcolor{Violet2}{\delta}$ }   (1)
		(1)  edge [in=110,out=150,loop, color=Navy] node[above] {$b:\textcolor{Violet2}{\bb}$ }   (1)
		(1)  edge [in=60,out=100,loop, color=Navy] node[above] {$c:\textcolor{Violet2}{\aa}$ }   (1)
		
		(1)  edge [out=210,in=250,loop, color=Green3] node[below] {$a:\textcolor{Violet2}{\ee}$ }   (1)
		(1)  edge [out=260,in=300,loop, color=Green3] node[below] {$b:\textcolor{Violet2}{\aa}$ }   (1)
		(1)  edge [in=210,out=-30, color=Green3] node[below] {$c:\textcolor{Violet2}{\gg}$ }   (2)
		
		(2)  edge [color=black] node[above] {$a:\textcolor{Violet2}{\gg}$ }   (1)
		(2)  edge [in=30,out=150, color=black] node[above] {$b:\textcolor{Violet2}{\aa}$ }   (1)
		(2)  edge [in=-20,out=20,loop, color=black] node[right] {$c:\textcolor{Violet2}{\zeta}$ }   (1);
		
		\end{tikzpicture}
		\caption{ The GFG Rabin automaton obtained from the "Zielonka tree" $\Z_{\F}$.}
		\label{Fig_GFGRabin}
	\end{figure}
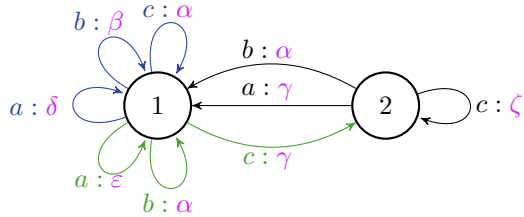
	The "Rabin condition" of this automaton is given by two "Rabin pairs" (corresponding to the "round nodes" of the "Zielonka tree" in Figure~\ref{Fig: ZielonkaTree-example}):
	
	\centering
	\begin{tabular}{c c}
		$G_\bb = \{\bb\}$, & $R_\bb= \{\aa, \gg, \ee, \zeta\}$,\\ 
		$G_\gg = \{\gg\}$, & $R_\gg= \{\aa, \bb, \delta\}$. 
	\end{tabular}
	
\end{example}

\subsubsection{Proof of correctness}

\begin{lemma}\label{lemma: acc-sequences-nodes-ZT}
	Let $w = n_0n_1n_2\dots \in N^\oo$ be an infinite sequence of nodes of the "Zielonka tree". The word $w$ satisfies the "Rabin condition" defined above if and only if there is a unique minimal node for the "ancestor relation" in $\minf(w)$ and this minimal node is "round" (recall that the root is the minimal element in $\zielonkatree_\F$).
\end{lemma}
\begin{proof}
	Suppose that there is a unique minimal node in $\minf(w)$, called $n$, and that $n$ is "round". We claim that $w$ is accepted by the "Rabin pair" $(G_n,R_n)$. It is clear that $\minf(w)\cap G_n\neq \emptyset$, because $n\in G_n$. It suffices to show that $\minf(w)\cap R_n = \emptyset$: By minimality, any other node $n'\in \minf(w)$ is a descendant of $n$ (equivalently, $n$ is an ancestor of $n'$), so $n'\notin R_n$.
	
	Conversely, suppose that $w\in N^\oo$ satisfies the "Rabin condition". Then, there is some "round node" $n\in \mroundnodes$ such that $\minf(w)\cap G_n \neq \emptyset$ and $\minf(w)\cap R_n = \emptyset$. Since $G_n = \{n\}$, we deduce that $n\in \minf(w)$. Moreover, as $\minf(w)\cap R_n = \emptyset$, all nodes in $\minf(w)$ are descendants of $n$. We conclude that $n$ is the unique minimal node in $\minf(w)$, and it is "round".
\end{proof}

\begin{lemma}\label{lemma: Rabin recognises L_F and GFG}
	The automaton $\RF$ recognises the language $"\L_\F"$ and is "good-for-games".
\end{lemma}

\begin{proof}
	$\L(\RF)\subseteq "\L_\F"$: Let $u\in \L(\RF)$ and let $w\in N^\oo$ be the sequence of nodes produced as output of an accepting run over $u$ in $\RF$. By \Cref{lemma: acc-sequences-nodes-ZT}, there is a unique  minimal node $n$ for $\ancestor$ appearing infinitely often in $w$ 
	and moreover $n$ is "round".
	Let $n_1, \dots, n_k$ be an enumeration of the children of $n$ (from left to right), with labels $\nu(n_i)\subseteq \GG$ (we remark that $\nu(n_i)\notin \F$, for $1\leq i \leq k$). We will prove that $\minf(u)\subseteq \nu(n)$ and $\minf(u)\nsubseteq \nu(n_i)$ for $1\leq i \leq k$. By definition of the "Zielonka tree", as $n$ is round, this implies that $\minf(u)\in \F$.

	Since eventually all nodes produced as output are descendants of $n$ (by minimality), $\minf(u)$ must be contained in $\nu(n)$ (by definition of the transitions of $\RF$).
	
	We suppose, towards a contradiction, that $\minf(u)\subseteq \nu(n_j)$ for some $1\leq j \leq k$.  
	Let $Q_i=\{\kl{\eta}(l)\: : \: l \text{ is a leaf below } n_i\}$ be the set of states corresponding to leaves under $n_i$, for $1\leq i \leq k$. We can suppose that the "leaves" "corresponding to" transitions of an accepting run over $u$ are all below $n$, and therefore, transitions of such a run only visit states in $\bigcup_{i=1}^kQ_i$. Indeed, eventually this is going to be the case, because if some of the leaves $l, l'$ "corresponding to" a transition $(q,a,n',q')$ are not below $n$, then $n'$ would not be a descendant of $n$ (since $n'$ is the least common ancestor of $l$ and $l'$).
	Also, by Property~\ref{Eq_property-star}, we have $Q_i\cap Q_j = \emptyset$, for all $i\neq j$.
	By definition of the transitions of $\RF$, if $c\in \GG$ is a colour in $\nu(n)$ but not in $\nu(n_i)$, all transitions from some state in $Q_i$ reading the colour $c$ go to $Q_{i+1}$, for $1\leq i \leq k-1$ (and to $Q_1$ if $i=k$).
	Also, if $c\in \nu(n_i)$, transitions from states in $Q_i$ reading $c$ stay in $Q_i$.	
	We deduce that a run over $u$ will eventually only visit states in $Q_j$, for some $j$ such that $\minf(u)\subseteq \nu(n_j)$. However, the only transitions
	from $Q_j$ that would produce $n$ as output are those corresponding to a colour $c\notin \nu(n_j)$, so the node $n$ is not produced infinitely often, a contradiction.

	\textbf{$"\L_\F"\subseteq \L(\RF)$ and good-for-gameness}: We will describe a strategy for a "resolver" in~$\RF$ using as memory the set of leaves of the Zielonka tree\footnote{This strategy is given by the (deterministic) Zielonka tree parity automaton $\P_\F$. It suffices to note that there is a morphism from $\P_\F$ to $\RF$ preserving all edges and the acceptance of loops.}. It will verify at every step that if the memory is on the leaf $l$, then $\RF$ is on the state $\kl{\eta}(l)$. The initial state of the memory is the leftmost leaf of $\zielonkatree_\muller$. If we are on the memory state $l$ and the letter $c\in\GG$ is read, we take the transition $(\kl{\eta}(l),c,n, \kl{\eta}(l'))\in \DD$, where $n$ is the maximal ancestor of $l$ such that $c\in \nu(n)$  and  $l'$ is the leftmost leaf in $\mjump_n(l)$; the memory state is updated to $l'$. Let us suppose that a word $u\in \L_\F$ is given as input to the automaton. We will see that the run produced by this strategy is accepting. We can suppose that the only colours appearing in $u$ are those of $\minf(u)$. Let $n$ be the leftmost $\ancestor$-maximal node such that $\minf(u)\subseteq \nu(n)$. Since $\minf(u)\in \F$, $n$ is a "round" node. We will prove that the run produced by the "resolver" above only produces nodes that are descendants of $n$ (including $n$) infinitely often and that it produces $n$ infinitely often and is therefore accepting. Let $n_1, \dots, n_k$ be the children of $n$ from left to right, and let $L_1, \dots, L_k$ be the (disjoint) sets of leaves below them, respectively. By the definition of the transitions and the strategy, the memory will eventually only consider "leaves" in $\bigcup_{i=1}^kL_i$, and will produced as output nodes that are "descendants" of $n$ (including $n$ itself). Also, each time that the memory is in some state in $L_i$ and a colour not in $\nu(n_i)$ is given, a transition leading to some state in $\kl{\eta}(L_{i+1})$ ($\kl{\eta}(L_{1})$ if $i=k$) producing the node $n$ as output is taken. Since $\minf(u)$ is not contained in $\nu(n_i)$ for $1\leq i \leq k$ (by the maximality assumption), this occurs infinitely often.
\end{proof}

\begin{remark}
	We have shown that given as input the "Zielonka tree" of a "Muller condition" $\F$ we can build in polynomial time a minimal "GFG" Rabin automaton for $"\L_\F"$. On the other hand, with the same input, it is $\NPc$ to decide whether there is a "deterministic" Rabin automaton of size $k$ recognising $\L_\F$ \cite[Theorem 31]{Casares2021Chromatic}. Therefore, unless $\mathtt{P}=\NP$, there are "Muller languages" for which minimal "deterministic" Rabin automata are strictly greater than minimal "GFG" Rabin automata. We will explicitly show some of these languages in Section~\ref{section:succinctness}.
\end{remark}

\subsubsection{Relation with the Zielonka-tree parity automaton}\label{subsubsection: relationWithZielonkaParity}

The "Zielonka tree" has been previously used to provide a minimal "deterministic" "parity automaton" for a "Muller condition"~\cite{CCF21Optimal, MeyerSickert21OptimalPractical}. The automata states are the "leaves" of the "Zielonka tree", and the transition from a leaf $l$ reading colour $c$ goes to the leftmost leave in $\mjump_n(l)$, where $n$ is the last ancestor of $l$ containing colour $c$. For example, \Cref{Fig_parity-automaton} shows a parity automaton recognising the "Muller condition" $\F = \{\{a,b\}, \{a,c\}, \{b\}\}$ from \Cref{example: Zielonka-tree}.

\begin{figure}[ht]
	\centering		
	\begin{tikzpicture}[square/.style={regular polygon,regular polygon sides=4}, align=center,node distance=2cm,inner sep=2pt]
	
	\node at (0,2) [state] (1navy) {$\delta$, \textcolor{Navy}{\textbf{1}}};
	\node at (3,2) [state] (1green) {$\ee$, \textcolor{Green3}{\textbf{1}}};
	\node at (6,2) [state] (2) {$\zeta$, \textbf{2}};
	
	\path[->] 
	(1navy)  edge [out=170,in=130,loop] node[left] {$a:0$ }   (1navy)
	(1navy)  edge [in=230,out=190,loop] node[left] {$b:1$ }   (1)
	(1navy)  edge [in=150,out=30] node[above] {$c:2$ }   (1green)
	
	(1green)  edge [out=110,in=70,loop] node[above] {$a:0$ }   (1green)
	(1green)  edge [in=-30,out=210] node[above] {$b:2$ }   (1navy)
	(1green)  edge [in=150,out=30] node[above] {$c:1$ }   (2)

	(2)  edge [in=-30,out=210] node[above] {$a:1$ }   (1green)
	(2)  edge [in=-40,out=-140, color=black] node[above] {$b:2$ }   (1navy)
	(2)  edge [out=110,in=70,loop] node[above] {$c:0$ }   (2);
	
	\end{tikzpicture}
	\caption{Parity automaton $\P_\F$ obtained from the Zielonka tree from \Cref{Fig: ZielonkaTree-example}.}
	\label{Fig_parity-automaton}
\end{figure}

This minimal parity automaton $\P_\F$ is closely related to the "GFG" "Rabin automaton" $\R_{\F}$
presented in \Cref{subsection:GFGforMuller}. More precisely, the automaton $\R_{\F}$ can be regarded as a quotient of $\P_\F$ given by the numbering $\kl{\eta}\colon \mleaves(\zielonkatree_\F) \to \{1,\dots,\memtree(\zielonkatree_\F)\}$. That is, to obtain $\R_{\F}$ we merge the states $\kl{\eta}^{-1}(i)$ for $1\leq i \leq \memtree(\zielonkatree_\F)$ and we keep all transitions. Moreover, the strategy for a resolver for $\R_{\F}$ as presented in the proof of \Cref{lemma: Rabin recognises L_F and GFG} is exactly given by the deterministic automaton $\P_\F$. However, we note that in general a "parity condition" is not sufficient in $\RF$ to accept $\L_\F$ and we need to replace it by a "Rabin" one.

We observe that the "GFG" "Rabin automaton" from \Cref{Fig_GFGRabin} is obtained as a quotient of the deterministic parity automaton in \Cref{Fig_parity-automaton}.

\subsubsection{Simplifications and optimisations}

\AP Given an automaton $\A=(Q, \Sigma, Q_0, \DD, \GG, \WW)$ we say that it has ""duplicated edges""  if there are some pair of states $q,q'\in Q$ and two different transitions between them labelled with the same input letter: $(q,a,\aa,q'),(q,a,\bb,q')\in \DD$.

As remarked previously, the construction we have presented provides an automaton potentially having "duplicated edges", which can be seen as an undesirable property (even if some automata models such as the HOA format~\cite{BBDKKMPS2015HOAFormat} allow them). We show next that we can always derive an equivalent automaton without "duplicated edges". Intuitively, in the Rabin case, if we want to merge two transitions having as output letters $\aa$ and $\bb$, we add a fresh letter $(\aa\bb)$ to label the new transition. For each "Rabin pair", this new letter will simulate the best of either $\aa$ or $\bb$ depending upon the situation.

\begin{proposition}[Simplification of automata]\label{prop: simplification_Rabin}
	Let $\A$ be a "Muller" (resp. "Rabin") automaton presenting "duplicated edges". There exists an equivalent "Muller" (resp. "Rabin") automaton $\A'$ on the same set of states without "duplicated edges". Moreover, if $\A$ is "GFG", $\A'$ can be chosen "GFG". In the Rabin case, the number of "Rabin pairs" is also preserved. 
\end{proposition}
\begin{proof}
	For the Rabin case, let $\A'$ be an automaton that is otherwise as $\A$ except that instead of the transitions $\Delta$  of $\A$ it only has one $a$-transition $q \xrightarrow{a:x} q'\in \Delta'$ (with a fresh colour $x$ per transition) per state-pair $q,q'$ and letter $a\in \Sigma$. That is, $\DD'=\{(q,a,x_j,q') \: : \: (q,a,y,q')\in \DD \text{ for some } y\}$. The new "Rabin condition" $\{(G_1',R_1'),\ldots, (G_r',R_r')\}$ is defined as follows. For each transition $q\xrightarrow{a:x} q'$:
	\begin{itemize}
		\item $x\in G_i'$ if $q\xrightarrow{a:y} q'\in \Delta$ for some $y\in G_i$ (there is a "green transition" for the $i^{th}$ pair)
		\item  $x\in R'_i$ if for all $q\xrightarrow{a:y} q'\in \Delta$,  $y\in R_i$ (there is no "green" or "orange transition" for the $i^{th}$ pair).
	\end{itemize}
	
	We claim that $L(\A')=L(\A)$. Indeed, if $u\in L(\A)$, as witnessed by some run $\rho$ and a "Rabin pair" $(G_i,R_i)$, then the corresponding run $\rho'$ in $\A'$ over $u$ is also accepting with Rabin pair $(G_i',R_i')$: the transitions of $\minf(\rho)\cap G_i$ induce transitions of $\minf(\rho')\cap G_i'$ and the fact that $\minf(\rho)\cap R_i=\emptyset$ guarantees that $\minf(\rho')\cap R_i'=\emptyset$.
	
	Conversely, if $u\in L(\A')$ as witnessed by a run $\rho'$ and Rabin pair $(G_i', R_i')$,
	then there is an accepting run $\rho$ over $u$ in $\A$: such a run can be obtained by choosing for each transition $q\xrightarrow{a:x}q'$ of $\rho'$ where $x\in G_i'$ a transition $q\xrightarrow{a:y}q'\in \Delta$ such that $y\in G_i$, which exists by definition of $\A'$, for each transition $q\xrightarrow{a:x}q'$ where $x\notin G_i\cup R_i$ a transition $q\xrightarrow{q,y}q'\in\Delta$ such that $y\notin R_i$, which also exists by definition of $\A'$, and for other transitions $q\xrightarrow{a:x}q'$ (that is, those for which $x\in R_i'$) an arbitrary transition $q\xrightarrow{a:y}q'\in \Delta$. Since $\rho'$ is accepting, we have $\minf(\rho')\cap G_i\neq \emptyset$ and $\minf(\rho)\cap R_i=\emptyset$, that is, $\rho$ is also accepting. 
	
	For the "Muller" case, the argument is even simpler. As above, we consider $\A'$ that is otherwise like $\A$ except that instead of the transitions $\Delta$ of $\A$, it only has one $a$-transition $q\xrightarrow{a:x}q'\in \Delta'$ (with a fresh colour per transition) per state-pair $q,q'$ and the accepting condition is defined as follows. A set of  transitions $T$ is accepting if and only if for each $t=q\xrightarrow{a:x}q'\in T$ there is a non-empty set $S_t \subseteq \{ q\xrightarrow{a:y}q'\in \Delta\}$ such that $\bigcup_{t\in T} S_t$ is accepting in $\A$. In other words, a set of transitions in $\A'$ is accepting if for each transition we can choose a non-empty subset of the original transitions in $\A$ that form an accepting run in $\A$.
	
	We claim that $L(\A')=L(\A)$. Indeed if $u\in L(\A)$, as witnessed by some run $\rho$, the run $\rho'$ that visits the same sequence of states in $\A'$ is accepting as witnessed by the transitions that occur infinitely often in $\rho$.
	
	Conversely, assume $u\in L(\A')$, as witnessed by a run $\rho'$ and a non-empty subset $S_t$ for each transitions $t$ that occurs infinitely often in $\rho'$ such that $\bigcup_{t\in \minf(\rho)} S_t$ is accepting in $\A$. Then there is an accepting run $\rho$ over $u$ in $\A$ that visits the same sequence of states as $\rho'$ and chooses instead of a transition $t\in \minf(\rho)$ each transition in $S_t$ infinitely often, and otherwise takes an arbitrary transition. The set of transitions $\rho$ visits infinitely often is exactly $\bigcup_{t\in\minf(\rho)} S_t$, and is therefore accepting.
	
	Finally, observe that in both cases, if $\A$ if "GFG", then the automaton $\A'$ without duplicate edges is also "GFG" since $\A'$ is obtained from $\A$ by merging transitions. Indeed, the "resolver" $r$ of $\A$ induces a "resolver" $r'$ for $\A'$ by outputting the unique transition with the same letter and state-pair as $r$. By the same argument as above, the run induced by $r'$ is accepting if and only if the run induced by $r$ is.
\end{proof}


\begin{example}
	The "GFG" "Rabin automaton" from Figure~\ref{Fig_GFGRabin} has "duplicated transitions". In Figure~\ref{Fig_GFGRabin-simplified} we present an equivalent "GFG" "Rabin automaton" without duplicates. For this, we have merged the self-loops in state $1$ labelled with $a$ and $b$ respectively. We have added the output letters $(\aa\bb)$ and $(\delta\ee)$. 
	The new "Rabin pairs" are given by:
	
	\centering
	\begin{tabular}{l l}
		$G_\bb' = \{\bb, (\aa\bb)\}$, & $R_\bb'= \{\aa, \gg, \ee, \zeta\}$,\\ 
		$G_\gg' = \{\gg\}$, & $R_\gg'= \{\aa, \bb, (\aa\bb), \dd\}$. 
	\end{tabular}
	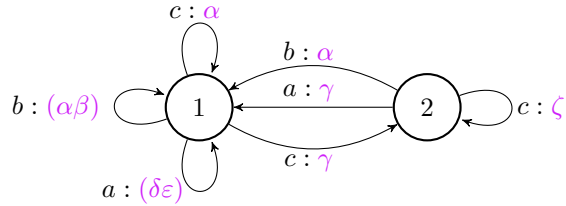
\begin{figure}[ht]
		\centering		
		\begin{tikzpicture}[square/.style={regular polygon,regular polygon sides=4}, align=center,node distance=2cm,inner sep=2pt]
		
		\node at (0,2) [state] (1) {$1$};
		\node at (3,2) [state] (2) {$2$};
		
		\path[->] 
		(1)  edge [out=250,in=290,loop] node[left] {$a:\textcolor{Violet2}{(\delta\ee)}$ }   (1)
		(1)  edge [in=160,out=200,loop] node[left] {$b:\textcolor{Violet2}{(\aa\bb)}$ }   (1)
		(1)  edge [in=70,out=110,loop] node[above] {$c:\textcolor{Violet2}{\aa}$ }   (1)
		
		(1)  edge [in=210,out=-30] node[below] {$c:\textcolor{Violet2}{\gg}$ }   (2)
		
		(2)  edge [color=black] node[above] {$a:\textcolor{Violet2}{\gg}$ }   (1)
		(2)  edge [in=30,out=150, color=black] node[above] {$b:\textcolor{Violet2}{\aa}$ }   (1)
		(2)  edge [in=-20,out=20,loop, color=black] node[right] {$c:\textcolor{Violet2}{\zeta}$ }   (1);
		
		\end{tikzpicture}
		\caption{ The simplified GFG Rabin automaton.}
		\label{Fig_GFGRabin-simplified}
		
	\end{figure}
	
\end{example}

\begin{remark}[Optimisation on the number of Rabin pairs]
	An important parameter in the study of Rabin automata is the number of "Rabin pairs" used. The automaton presented in this Section uses a number of Rabin pairs that equals the number of "round nodes" in the "Zielonka tree". This can be improved by using only the "round nodes" in the Zielonka directed acyclic graph (obtained from the tree by merging nodes with the same labels). However, even this latter option is not always optimal and we conjecture that minimising the number of Rabin pairs without increasing the size of the automaton is $\NPc$.
\end{remark}


	
	\section{GFG Rabin automata recognising Muller conditions can be exponentially more succinct than deterministic ones in number of states}
	\label{section:succinctness}

On his PhD Thesis \cite{Kopczynski2006Half,Kopczynski2008PhD}, Kopczyński raised the question of whether the "general" and the "chromatic" memory requirements of winning conditions always coincide. By \Cref{Th_ GFGRabin=Memory} and \cite[Theorem 28]{Casares2021Chromatic}, in the case of Muller conditions, this question is equivalent to the following:

\begin{quote}
	Is there a "Muller language" $L$ such that minimal "GFG Rabin automata" recognising $L$ are strictly smaller than "deterministic" Rabin automata for $L$?
\end{quote}
In \cite{Casares2021Chromatic} this question is answered positively. It is shown that for every $n\in\NN$ there is a "Muller language" $L_n$ over an alphabet $\GG_n$ such that a minimal GFG Rabin automaton for it has size $2$, but a minimal deterministic Rabin automaton for it has size $n$. However, the size of the alphabet $\GG_n$ in that example also has size $n$. A natural question is whether GFG Rabin automata recognising "Muller conditions" can be exponentially more succinct than deterministic ones, when also taking into account the alphabet size. This is indeed the case:

\begin{theorem}\label{Th: GFGRabin-Exp-Succinct}
		There exists a constant $\aa>1$, a sequence of natural numbers $n_1<n_2<n_3\dots$ and a sequence of Muller conditions $\F_{n_i}$ over $\GG_{n_i}=\{1,\dots, n_i\}$ such that 
		\begin{itemize}
			\item a minimal "GFG" "Rabin automaton" for $\L_{\F_{n_i}}$ has size $\lfloor n_i/2 \rfloor$,
			\item a minimal deterministic Rabin automaton for $\L_{\F_{n_i}}$ has size at least $\aa^{n_i}$.
		\end{itemize} 
	A lower bound for such a constant is $\aa= 1.116$.
\end{theorem}

%

We devote the rest of this Section to proving Theorem~\ref{Th: GFGRabin-Exp-Succinct}. In brief, the "Muller conditions" in question require half the colours to be seen infinitely often. The construction of the small "GFG Rabin automaton" follows from constructing the "Zielonka tree" of the condition. For the lower bound on the deterministic Rabin automaton, we reduce the problem to finding a lower bound on the "chromatic number" of a certain graph, which we finally show to be sufficiently large for a family of our Muller conditions.

Let $n\in\NN$. We define the following Muller condition over $\GG_n=\{1,\dots,n\}$:
\[ \F_n=\{ C\subseteq \GG_n \: : \: |C|= \lfloor n/2 \rfloor\}. \]
The Zielonka tree of $\F_n$ is depicted in Figure~\ref{Fig: ZielonkaTree-F_n} (for $n$ even).

\begin{figure}[ht]
	\centering
	\begin{tikzpicture}[square/.style={regular polygon,regular polygon sides=4}, align=center,node distance=2cm,inner sep=3pt]
	
	\node at (0,2.6) [draw, rectangle, minimum height=0.9cm, minimum width = 1.5cm] (R) {$1,2, \dots, n$};
	
	\node at (-5,1.3) [draw, ellipse,minimum height=0.8cm, minimum width = 1.3cm, scale=0.9] (0) {$1,2,\dots, \frac{n}{2}$};
	\node at (-1,1.3) [draw, ellipse, minimum height=0.8cm, minimum width = 1.3cm, scale=0.9] (1) {$1,3,\dots, , \frac{n}{2}+1 $};
	\node at (1.5,1.3)  (dots1) {$\dots$};
	\node at (4,1.3) [draw, ellipse, minimum height=0.8cm, minimum width = 1.3cm, scale=0.9] (2) {$\frac{n}{2}, \dots, n$};

	\node at (-6.5,0) [draw, rectangle,minimum width=0.9cm, scale=0.8] (00) {$1,\dots, \frac{n}{2}-1$};
	\node at (-5.1,0)  (dots1) {$\dots$};
	\node at (-4,0) [draw, rectangle,minimum width=0.9cm, scale=0.8] (01) {$2,\dots, \frac{n}{2}$};
	
	\node at (-2.2,0) [draw, rectangle,minimum width=0.9cm, scale=0.8] (10) {$1,3, \dots, \frac{n}{2}$};
	\node at (-1,0)  (dots1) {$\dots$};
	\node at (0.2,0) [draw, rectangle,minimum width=0.9cm, scale=0.8] (11) {$3,\dots, \frac{n}{2}+1$};
	
	
	\node at (2.9,0) [draw, rectangle,minimum width=0.9cm, scale=0.8] (20) {$\frac{n}{2}, \dots, n-1$};
	\node at (4.2,0)  (dots1) {$\dots$};
	\node at (5.5,0) [draw, rectangle,minimum width=0.9cm, scale=0.8] (21) {$\frac{n}{2}+1,\dots, n$};

	\draw   
	(R) edge (0)
	(R) edge (1)
	(R) edge (2)
	
	(0) edge (00)
	(0) edge (01)
	(1) edge (10)
	(1) edge (11)
	(2) edge (20)
	(2) edge (21);
	
	\end{tikzpicture}
	\caption{Zielonka tree $\zielonkatree_{\F_{n}}$ for 
		$\F_n=\{ C\subseteq \GG_n \: : \: |C|=\lfloor n/2 \rfloor\}$.}
	\label{Fig: ZielonkaTree-F_n}
\end{figure}

Each round node in $\zielonkatree_{\F_{n}}$ has exactly $\lfloor n/2 \rfloor$ children, and therefore $\memtree(\zielonkatree_{\F_{n_i}})=\lfloor n/2 \rfloor$. Thus, a minimal GFG Rabin automaton recognising $"\L_{\F_n}"$ has size $\lfloor n/2 \rfloor$ (by \Cref{prop: optimalMemoryDJW} and \Cref{Th_ GFGRabin=Memory}).

We now give a lower bound for deterministic Rabin automata recognising $"\L_{\F_n}"$. Our main tool will be Lemma~\ref{lemma:Rabin-Cycles}, which uses the notion of \emph{cycles}.
\AP A ""cycle"" of an automaton $\A$ is a set of transitions forming a closed path (not necessarily simple). The set of states of a "cycle" consists of those states that are the source of some transition in it. If $\A$ is a "Rabin automaton", we say that a "cycle" is accepting (resp. rejecting) if the colours $c_1,\dots,c_k$ appearing in its transitions form a word $(c_1c_2\dots c_k)^\omega$ that satisfies (resp. does not satisfy) the "Rabin condition".

\begin{lemma}[\cite{CCF21Optimal}]\label{lemma:Rabin-Cycles}
	Let $\A$ be a "deterministic" "Rabin automaton". If $\ell_1$ and $\ell_2$ are two rejecting "cycles" in $\A$ with some state in common, then the union of $\ell_1$ and $\ell_2$ is also a rejecting "cycle".
\end{lemma}

For the following, let $\A$ be a "deterministic" "Rabin automaton" recognising $\L_{\F_n}$. \AP For each subset of letters $C\subseteq \GG_n$, we define a ""final $C$-Strongly Connected Component"" (\reintro{$C$-FSCC} for short) as a set of states $P$ of $\A$ such that:
\begin{itemize}
	\item For every pair of states $p,q\in P$, there is a word $w\in C^*$ labelling a path from $p$ to $q$.
	\item For every $p\in P$ and $w\in C^*$, the run over $w$ starting in $p$ remains in $P$.
\end{itemize}

It is easy to see that for every $C\subseteq \GG_n$ there exists some "$C$-FSCC" in $\A$.

\begin{lemma}\label{lemma: Disjoint-SCC}
	Let $C_1,C_2\subseteq \GG_n$ such that $|C_i|< \lfloor n/2 \rfloor$, for $i=1,2$ and such that $|C_1 \cup C_2|=\lfloor n/2 \rfloor$. If $P_1$ and $P_2$ are two $C_1$ and $C_2$-FSCC, respectively, then $P_1\cap P_2=\emptyset$.
\end{lemma}
\begin{proof}
	For $i=1,2$, let $\ell_i$ be a "cycle" visiting all states of $P_i$ and reading exactly the set of letters $C_i$. By definition of $\L_{\F_{n}}$, $\ell_i$ is a rejecting "cycle". If $P_1$ and $P_2$ had some state in common, we could take the union of the cycles $\ell_1$ and $\ell_2$, producing an accepting "cycle", which is impossible by Lemma~\ref{lemma:Rabin-Cycles}.
\end{proof}

\noindent We associate the following (undirected) graph $""\G_{\F_n}""=(V_{\F_n}, E_{\F_n})$ to the "Muller condition" $\F_n$:
\begin{itemize}
	\item $V_{\F_n}= \P(\GG_n)$.
	\item There is an edge between two subsets $C_1,C_2\subseteq \GG_n$ if and only if $|C_i|<\lfloor n/2 \rfloor$, for $i=1,2$, and $|C_1 \cup C_2|= \lfloor n/2 \rfloor$.
\end{itemize}

That is, we connect two vertices if they correspond to rejecting sets but taking their union we obtain an accepting set.

We reduce finding lower bounds in the size of deterministic Rabin automata to giving lower bounds for the "chromatic number" of $"\G_{\F_n}"$. 
\AP A ""colouring"" of an undirected graph $G=(V,E\subseteq V\times V)$ is a mapping $c:V\rightarrow \Lambda$ such that $c(v)=c(v')\Rightarrow (v,v')\notin E$ for every pair of nodes $v, v' \in V$. We say that such a colouring has size $|\Lambda|$.
\AP The ""chromatic number"" of $G$ is the minimal number $k$ such that $G$ has a "colouring" of size $k$. We denote it $""\chi""(G)$.

\begin{lemma}\label{lemma: chromatic-lower-bound}
	A lower bound for the size of a minimal deterministic Rabin automaton recognising $\L_{\F_n}$ is given by $\kl{\chi}("\G_{\F_n}")$.
\end{lemma}

\begin{proof}
	Let $\A$ be a deterministic Rabin automaton recognising $\L_{\F_n}$ with states $Q$. We define a "colouring" $c$ of $"\G_{\F_n}"$ using $Q$ as colours. For each $C\subseteq \GG_n$, we let $P_C$ be a "$C$-FSCC" and we pick a state $q_C\in P_C$. We define $c(C)=q_C$. We prove that this is a correct colouring. Suppose that $C_1$ and $C_2$ are two vertices in $\G_{\F_n}$ connected by some edge, that is, $|C_i|<\lfloor n/2 \rfloor$ and $|C_1 \cup C_2|=\lfloor n/2 \rfloor$. If $q_{C_1}=q_{C_2}$, it means that $P_{C_1}\cap P_{C_2}\neq \emptyset$, contradicting Lemma~\ref{lemma: Disjoint-SCC}.
\end{proof}

\begin{remark}
		The definition of $"\G_{\F_n}"$ is not specific to this "Muller condition". It can be defined analogously for any other "Muller condition" and Lemma~\ref{lemma: chromatic-lower-bound} holds by the same argument.
	\end{remark}


\begin{proposition}\label{prop: graphTheory}
There exists a constant $\aa>1$ and a sequence of natural numbers $n_1<n_2<n_3\dots$ such that  $\aa^{n_i} \; \leq \; \kl{\chi}(\G_{\F_{n_i}}).$
\end{proposition}


In order to prove \Cref{prop: graphTheory} we introduce some further graph-theoretic notions. Let $\G=(V,E)$ be an undirected graph. 
 An \AP""independent set"" of $\G$ is a set $S\subseteq V$ such that $(v,v')\notin E$ for every pair of vertices $v,v'\in S$.

\begin{lemma}\label{lemma: subset-chromaticNumber}
	Let $R\subseteq V$, and let $\G_R=(R, E|_{R\times R})$ be the subgraph of $\G$ induced by $R$. Then, $\kl{\chi}(\G)\geq \kl{\chi}(\G_R)$.
\end{lemma}

\begin{lemma}\label{Lemma: indset}
	Let $m$ be an upper bound on the size of the independent sets in $\G$. Then
	\[ \kl{\chi}(\G)\geq \dfrac{|V|}{m}.\]
\end{lemma}
\begin{proof}
	Let $c\colon V \to \Lambda$ be a "colouring" of $\G$ with $|\Lambda|=\kl{\chi}(\G)$. Then, by definition of a "colouring", for each $x\in \Lambda$, $c^{-1}(x)$ is an "independent set" in $\G$, so $|c^{-1}(x)| \leq m$. Also, $V = \bigcup_{x\in \Lambda}c^{-1}(x)$, so 
	\[ |V| = \sum\limits_{x\in \Lambda} |c^{-1}(x)| \leq \kl{\chi}(\G)\cdot m. \qedhere\]
\end{proof}

We will find a subgraph of $\G_{\F_n}$ for which we can provide an upper bound on the size of its independent sets. The upper bound is provided by the following theorem (adapted from \cite[Theorem 15]{MubayiRodl2014Intersections}).
\begin{theorem}[\cite{MubayiRodl2014Intersections}, Theorem 15]\label{Th: Chromatic-Marthe}
	Let $n>k>2t$ such that $k-t$ is a prime number. Suppose that $\B$ is a family of subsets of size $k$ of $\GG_n$ such that $|A\cap B|\neq t$ for any pair of subsets $A,B\in \B$. Then,
	\[ |\B| \leq {n \choose k-t-1}. \]
\end{theorem}

We conclude this section with the proof of \Cref{prop: graphTheory}.
\begin{proof}[Proof of \Cref{prop: graphTheory}]
	Let $p$ be a prime number and let $n=5p$.
	We will study the subgraph of $\G_{\F_n}$ formed by the subsets of size exactly $k=\lfloor 3n/10 \rfloor$. We denote this subgraph by $H_{n,k}$. Two subsets $A, B\subseteq \GG_n$ of size $k$ verify that $|A\cup B|= \lfloor n/2 \rfloor$ if and only if $|A\cap B|=\lfloor n/10 \rfloor$. We set $t=\lfloor n/10 \rfloor$.
	We get $k-t=p$ so we can apply  Theorem~\ref{Th: Chromatic-Marthe} and we obtain that any "independent set" in $H_{n,k}$  has size at most ${n \choose \frac{1}{5}n-1}$. By Lemma~\ref{Lemma: indset}, $\kl{\chi}(H_{n,k})\geq {n \choose \lfloor\frac{3}{10}n\rfloor}/{n \choose \frac{1}{5}n-1}$. By Lemma~\ref{lemma: subset-chromaticNumber} we know that this lower bound also holds for $\G_{\F_n}$. Using Stirling's approximation we obtain that
	\[ \kl{\chi}(\G_{\F_n}) \geq {n \choose \lfloor\frac{3}{10}n\rfloor}/{n \choose \frac{1}{5}n-1} = \Omega\left(\left( \dfrac{(1/5)^{1/5}(4/5)^{4/5}}{(3/10)^{3/10}(7/10)^{7/10}}\right)^n\right) = \Omega(1.116^n).   \]
	
	To conclude, we take an enumeration of prime numbers, $p_1<p_2<\dots$ and we set $n_i=5p_i$.  	
\end{proof}

\begin{remark}[Choices of $k$ and $t$]
	The choice of $k=\lfloor 3n/10 \rfloor$ and $t=\lfloor n/10\rfloor$ in the previous proof might appear quite enigmatic.  We try to explain them now.
	
	We want to find a number $k$ such that there is not a big family of sets $\{A_i\subseteq \GG_n\}$ of size $|A_i|=k$ such that $|A_i\cup A_j|\neq n/2$, and express this fact in terms of $|A_i\cap A_j|$.
	Since $|A\cup B|=2k - |A\cap B|$, if we define $t=2k-n/2$, then $|A\cup B|\neq n/2$ if and only if $|A\cap B|\neq t$, so the value of $t$ will be completely determined by the choice of $k$.
	Our objective is to minimise the upper bound given in \Cref{Th: Chromatic-Marthe} (what we do by minimising $k-t$) while making sure that the hypothesis $k>2t$ is verified. In the boundary of this condition ($k=2t$) we obtain $k=n/3$, so we express our choices as $k=(1/3-\ee)n$ and $t=(1/6-2\ee)n$. Moreover, $k-t=(1/6+\ee)n$ has to be a prime number (for infinite $n$). If $1/6+\ee=1/q$ for some $q\in\NN$, we would succeed by considering $n$ of the form $q\cdot p$, for $p$ a prime number. We will therefore take $\ee=\frac{6-q}{6q}$, for some $q$, $1\leq q \leq 5$. With the optimal choice, $q=5$, we obtain $k=3n/10$, $t=n/10$ and $k-t=n/5$. Since $k$ and $t$ will not be integers for $n$ of the form $5p$ ($p$ a prime number) we are forced to take the integer part in the proof of \Cref{prop: graphTheory}.
\end{remark}




	\section{Conclusion}
	\label{section:conclusion}

We believe that our work is a significative advance in the understanding of the memory needed for winning $\omega$-regular games.
In combination with the literature, we can describe the current understanding of "Muller languages" as follows: 
\begin{itemize}
\item The least "memory" necessary for winning all won
	$L$-games equals the least number of states of a "GFG Rabin automaton" for~$L$.
\item Computing this quantity can be done in polynomial time for $L$ given by its Zielonka tree.
\item The least "chromatic memory" necessary for winning all won
	$L$-games equals the least number of states of a "deterministic Rabin automaton" for~$L$.
\item Computing this quantity is $\NPc$ for~$L$ given by its Zielonka tree.
\item The "chromatic memory" can be arbitrarily larger than the "memory".
	It can be exponential in the size the alphabet, even while the "memory" remains linear.
\end{itemize}
This description shows that "GFG automata" play a key, and, up till now, unexplored role in understanding the complexity of "Muller languages" and that this role is ---in some respect---even more important than that of the more classical deterministic automata.


\bibliographystyle{plainurl} 
\bibliography{bibMinimisation}

\begin{thebibliography}{10}

\bibitem{AbuRadiKupferman19Minimizing}
Bader {Abu Radi} and Orna Kupferman.
\newblock Minimizing {GFG} transition-based automata.
\newblock In {\em ICALP}, volume 132, pages 100:1--100:16, 2019.
\newblock \href {https://doi.org/10.4230/LIPIcs.ICALP.2019.100}
  {\path{doi:10.4230/LIPIcs.ICALP.2019.100}}.

\bibitem{BBDKKMPS2015HOAFormat}
Tom{\'a}{\v{s}} Babiak, Franti{\v{s}}ek Blahoudek, Alexandre Duret-Lutz,
  Joachim Klein, Jan K{\v{r}}et{\'i}nsk{\'y}, David M{\"u}ller, David Parker,
  and Jan Strej{\v{c}}ek.
\newblock The {H}anoi omega-automata format.
\newblock In {\em CAV}, pages 479--486, 2015.

\bibitem{BK18}
Marc Bagnol and Denis Kuperberg.
\newblock B{\"u}chi good-for-games automata are efficiently recognizable.
\newblock In {\em FSTTCS}, page~16, 2018.

\bibitem{BKLS20}
Udi Boker, Denis Kuperberg, Karoliina Lehtinen, and Michal Skrzypczak.
\newblock On succinctness and recognisability of alternating good-for-games
  automata.
\newblock {\em CoRR}, abs/2002.07278, 2020.
\newblock \href {http://arxiv.org/abs/2002.07278} {\path{arXiv:2002.07278}}.

\bibitem{BL21a}
Udi Boker and Karoliina Lehtinen.
\newblock History determinism vs. good for gameness in quantitative automata,
  2021.
\newblock \href {http://arxiv.org/abs/2110.14238} {\path{arXiv:2110.14238}}.

\bibitem{BL22}
Udi Boker and Karoliina Lehtinen.
\newblock Token games and history-deterministic quantitative-automata, 2022.
\newblock To appear in proceedings of FoSSaCS'22.
\newblock \href {http://arxiv.org/abs/2110.14308} {\path{arXiv:2110.14308}}.

\bibitem{BRV22OmegaRegMemory}
Patricia Bouyer, Mickael Randour, and Pierre Vandenhove.
\newblock Characterizing omega-regularity through finite-memory determinacy of
  games on infinite graphs.
\newblock {\em CoRR}, abs/2110.01276, 2021.
\newblock \href {http://arxiv.org/abs/2110.01276} {\path{arXiv:2110.01276}}.

\bibitem{BRORV20FiniteMemory}
Patricia Bouyer, St{\'{e}}phane~Le Roux, Youssouf Oualhadj, Mickael Randour,
  and Pierre Vandenhove.
\newblock Games where you can play optimally with arena-independent finite
  memory.
\newblock In {\em CONCUR}, volume 171, pages 24:1--24:22, 2020.
\newblock \href {https://doi.org/10.4230/LIPIcs.CONCUR.2020.24}
  {\path{doi:10.4230/LIPIcs.CONCUR.2020.24}}.

\bibitem{Buchi77Games}
J.~Richard B{\"{u}}chi.
\newblock Using determinancy of games to eliminate quantifiers.
\newblock In {\em FCT}, volume~56 of {\em Lecture Notes in Computer Science},
  pages 367--378. Springer, 1977.
\newblock \href {https://doi.org/10.1007/3-540-08442-8\_104}
  {\path{doi:10.1007/3-540-08442-8\_104}}.

\bibitem{Casares2021Chromatic}
Antonio Casares.
\newblock On the minimisation of transition-based {R}abin automata and the
  chromatic memory requirements of {M}uller conditions.
\newblock In {\em CSL}, volume 216, pages 12:1--12:17, 2022.
\newblock \href {https://doi.org/10.4230/LIPIcs.CSL.2022.12}
  {\path{doi:10.4230/LIPIcs.CSL.2022.12}}.

\bibitem{CCF21Optimal}
Antonio Casares, Thomas Colcombet, and Nathana\"{e}l Fijalkow.
\newblock Optimal transformations of games and automata using {M}uller
  conditions.
\newblock In {\em ICALP}, volume 198, pages 123:1--123:14, 2021.
\newblock \href {https://doi.org/10.4230/LIPIcs.ICALP.2021.123}
  {\path{doi:10.4230/LIPIcs.ICALP.2021.123}}.

\bibitem{Colcombet2009CostFunctions}
Thomas Colcombet.
\newblock The theory of stabilisation monoids and regular cost functions.
\newblock In {\em ICALP}, pages 139--150, 2009.
\newblock \href {https://doi.org/10.1007/978-3-642-02930-1\_12}
  {\path{doi:10.1007/978-3-642-02930-1\_12}}.

\bibitem{ColcombetN2006PositionalEdge}
Thomas Colcombet and Damian Niwiński.
\newblock On the positional determinacy of edge-labeled games.
\newblock {\em Theoretical Computer Science}, 352(1):190--196, 2006.
\newblock \href {https://doi.org/https://doi.org/10.1016/j.tcs.2005.10.046}
  {\path{doi:https://doi.org/10.1016/j.tcs.2005.10.046}}.

\bibitem{Colcombetz2009tight}
Thomas Colcombet and Konrad Zdanowski.
\newblock A tight lower bound for determinization of transition labeled
  {B}{\"u}chi automata.
\newblock In {\em ICALP}, pages 151--162, 2009.
\newblock \href {https://doi.org/10.1007/978-3-642-02930-1\_13}
  {\path{doi:10.1007/978-3-642-02930-1\_13}}.

\bibitem{DJW1997memory}
Stefan Dziembowski, Marcin Jurdzi{ń}ski, and Igor Walukiewicz.
\newblock How much memory is needed to win infinite games?
\newblock In {\em LICS}, pages 99--110, 1997.
\newblock \href {https://doi.org/10.1109/LICS.1997.614939}
  {\path{doi:10.1109/LICS.1997.614939}}.

\bibitem{EmersonJutla91Determinacy}
E.~Allen Emerson and Charanjit~S. Jutla.
\newblock Tree automata, mu-calculus and determinacy (extended abstract).
\newblock In {\em FOCS}, pages 368--377, 1991.
\newblock \href {https://doi.org/10.1109/SFCS.1991.185392}
  {\path{doi:10.1109/SFCS.1991.185392}}.

\bibitem{GimbertZielonka2005Memory}
Hugo Gimbert and Wieslaw Zielonka.
\newblock Games where you can play optimally without any memory.
\newblock In {\em CONCUR}, volume 3653, pages 428--442, 2005.
\newblock \href {https://doi.org/10.1007/11539452\_33}
  {\path{doi:10.1007/11539452\_33}}.

\bibitem{GJLZ21}
Shibashis Guha, Isma\"{e}l Jecker, Karoliina Lehtinen, and Martin Zimmermann.
\newblock {A Bit of Nondeterminism Makes Pushdown Automata Expressive and
  Succinct}.
\newblock In {\em MFCS}, volume 202, pages 53:1--53:20, 2021.
\newblock \href {https://doi.org/10.4230/LIPIcs.MFCS.2021.53}
  {\path{doi:10.4230/LIPIcs.MFCS.2021.53}}.

\bibitem{Gurevich1982trees}
Yuri Gurevich and Leo Harrington.
\newblock Trees, automata, and games.
\newblock In {\em STOC}, pages 60--65, 1982.
\newblock \href {https://doi.org/10.1145/800070.802177}
  {\path{doi:10.1145/800070.802177}}.

\bibitem{HP06}
Thomas~A. Henzinger and Nir Piterman.
\newblock Solving games without determinization.
\newblock In {\em Computer Science Logic}, pages 395--410, 2006.

\bibitem{Horn09RandomFruits}
Florian Horn.
\newblock Random fruits on the zielonka tree.
\newblock In {\em STACS}, volume~3, pages 541--552, 2009.
\newblock \href {https://doi.org/10.4230/LIPIcs.STACS.2009.1848}
  {\path{doi:10.4230/LIPIcs.STACS.2009.1848}}.

\bibitem{Klarlund94Determinacy}
Nils Klarlund.
\newblock Progress measures, immediate determinacy, and a subset construction
  for tree automata.
\newblock {\em Annals of Pure and Applied Logic}, 69(2):243--268, 1994.
\newblock \href {https://doi.org/https://doi.org/10.1016/0168-0072(94)90086-8}
  {\path{doi:https://doi.org/10.1016/0168-0072(94)90086-8}}.

\bibitem{Kopczynski2006Half}
Eryk Kopczy{\'n}ski.
\newblock Half-positional determinacy of infinite games.
\newblock In {\em ICALP}, pages 336--347, 2006.
\newblock \href {https://doi.org/10.1007/11787006\_29}
  {\path{doi:10.1007/11787006\_29}}.

\bibitem{Kopczynski2008PhD}
Eryk Kopczy{\'n}ski.
\newblock Half-positional determinacy of infite games. {P}h{D} {T}hesis.
\newblock 2008.

\bibitem{Kozachinskiy22Chromatic}
Alexander Kozachinskiy.
\newblock State complexity of chromatic memory in infinite-duration games.
\newblock {\em CoRR}, abs/2201.09297, 2022.
\newblock \href {http://arxiv.org/abs/2201.09297} {\path{arXiv:2201.09297}}.

\bibitem{Kretinski2017IAR}
Jan K{\v{r}}et{\'i}nsk{\'y}, Tobias Meggendorfer, Clara Waldmann, and
  Maximilian Weininger.
\newblock Index appearance record for transforming {R}abin automata into parity
  automata.
\newblock In {\em TACAS}, pages 443--460, 2017.
\newblock \href {https://doi.org/10.1007/978-3-662-54577-5\_26}
  {\path{doi:10.1007/978-3-662-54577-5\_26}}.

\bibitem{KS15}
Denis Kuperberg and Michał Skrzypczak.
\newblock On determinisation of good-for-games automata.
\newblock In {\em ICALP}, pages 299--310, 2015.
\newblock \href {https://doi.org/10.1007/978-3-662-47666-6_24}
  {\path{doi:10.1007/978-3-662-47666-6_24}}.

\bibitem{LZ20}
Karoliina Lehtinen and Martin Zimmermann.
\newblock Good-for-games $\omega$-pushdown automata.
\newblock In {\em LICS}, page 689–702, 2020.
\newblock \href {https://doi.org/10.1145/3373718.3394737}
  {\path{doi:10.1145/3373718.3394737}}.

\bibitem{LodingP19}
Christof L{\"{o}}ding and Anton Pirogov.
\newblock Determinization of {B}{\"{u}}chi automata: Unifying the approaches of
  {S}afra and {M}uller-{S}chupp.
\newblock In {\em ICALP}, pages 120:1--120:13, 2019.
\newblock \href {https://doi.org/10.4230/LIPIcs.ICALP.2019.120}
  {\path{doi:10.4230/LIPIcs.ICALP.2019.120}}.

\bibitem{LMS20SynthesisLTL}
Michael Luttenberger, Philipp~J. Meyer, and Salomon Sickert.
\newblock Practical synthesis of reactive systems from {LTL} specifications via
  parity games.
\newblock {\em Acta Informatica}, pages 3--36, 2020.
\newblock \href {https://doi.org/10.1007/s00236-019-00349-3}
  {\path{doi:10.1007/s00236-019-00349-3}}.

\bibitem{Loding1999Optimal}
Christof Löding.
\newblock Optimal bounds for transformations of $\omega$-automata.
\newblock In {\em FSTTCS}, page 97–109, 1999.
\newblock \href {https://doi.org/10.1007/3-540-46691-6\_8}
  {\path{doi:10.1007/3-540-46691-6\_8}}.

\bibitem{McNaughton1966Testing}
Robert McNaughton.
\newblock Testing and generating infinite sequences by a finite automaton.
\newblock {\em Information and control}, 9:521--530, 1966.

\bibitem{MeyerSickert21OptimalPractical}
Philipp Meyer and Salomon Sickert.
\newblock On the optimal and practical conversion of {E}merson-{L}ei automata
  into parity automata.
\newblock {\em Personal Communication}, 2021.

\bibitem{Michel1988Complementation}
Max Michel.
\newblock Complementation is more difficult with automata on infinite words.
\newblock {\em CNET, Paris}, 15, 1988.

\bibitem{MubayiRodl2014Intersections}
Dhruv Mubayi and Vojtech R{ö}dl.
\newblock Specified intersections.
\newblock {\em Transactions of the American Mathematical Society},
  366(1):491--504, 2014.
\newblock URL: \url{http://www.jstor.org/stable/23813142}.

\bibitem{MullerSchupp95NewResults}
David~E. Muller and Paul~E. Schupp.
\newblock Simulating alternating tree automata by nondeterministic automata:
  New results and new proofs of the theorems of {R}abin, {M}c{N}aughton and
  {S}afra.
\newblock {\em Theor. Comput. Sci.}, 141(1–2):69–107, 1995.
\newblock \href {https://doi.org/10.1016/0304-3975(94)00214-4}
  {\path{doi:10.1016/0304-3975(94)00214-4}}.

\bibitem{Piterman2006fromNDBuchi}
Nir Piterman.
\newblock From nondeterministic {B}{\"u}chi and {S}treett automata to
  deterministic parity automata.
\newblock In {\em LICS}, pages 255--264, 2006.
\newblock \href {https://doi.org/10.1109/LICS.2006.28}
  {\path{doi:10.1109/LICS.2006.28}}.

\bibitem{PR89Synthesis}
Amir Pnueli and Roni Rosner.
\newblock On the synthesis of a reactive module.
\newblock In {\em POPL}, page 179–190, 1989.
\newblock \href {https://doi.org/10.1145/75277.75293}
  {\path{doi:10.1145/75277.75293}}.

\bibitem{Safra1988onthecomplexity}
Schmuel Safra.
\newblock On the complexity of $\omega$-automata.
\newblock In {\em FOCS}, page 319–327, 1988.
\newblock \href {https://doi.org/10.1109/SFCS.1988.21948}
  {\path{doi:10.1109/SFCS.1988.21948}}.

\bibitem{Schewe2009tighter}
Sven Schewe.
\newblock Tighter bounds for the determinisation of {B}{\"u}chi automata.
\newblock In {\em FoSSaCS}, pages 167--181, 2009.
\newblock \href {https://doi.org/10.1007/978-3-642-00596-1\_13}
  {\path{doi:10.1007/978-3-642-00596-1\_13}}.

\bibitem{Schewe10MinimisingNPComplete}
Sven Schewe.
\newblock Beyond hyper-minimisation---minimising {DBA}s and {DPA}s is
  {NP}-complete.
\newblock In {\em FSTTCS}, volume~8, pages 400--411, 2010.
\newblock \href {https://doi.org/10.4230/LIPIcs.FSTTCS.2010.400}
  {\path{doi:10.4230/LIPIcs.FSTTCS.2010.400}}.

\bibitem{Schewe20MinimisingGFG}
Sven Schewe.
\newblock Minimising {G}ood-{F}or-{G}ames automata is {NP}-complete.
\newblock In {\em FSTTCS}, volume 182, pages 56:1--56:13, 2020.
\newblock \href {https://doi.org/10.4230/LIPIcs.FSTTCS.2020.56}
  {\path{doi:10.4230/LIPIcs.FSTTCS.2020.56}}.

\bibitem{Zielonka1998infinite}
Wies{\l}aw Zielonka.
\newblock Infinite games on finitely coloured graphs with applications to
  automata on infinite trees.
\newblock {\em Theoretical Computer Science}, 200(1-2):135--183, 1998.
\newblock \href {https://doi.org/10.1016/S0304-3975(98)00009-7}
  {\path{doi:10.1016/S0304-3975(98)00009-7}}.

\end{thebibliography}
	
\end{document}